\documentclass[%
 reprint,
superscriptaddress,
 amsmath,amssymb,
 aps,
prx,
floatfix,
]{revtex4-2}

\usepackage{graphicx}
\usepackage{dcolumn}
\usepackage{bm}

\usepackage{hyperref}
\usepackage{booktabs}
\usepackage[utf8]{inputenc}
\usepackage{amsfonts}
\usepackage{amsthm}
\usepackage{algorithmic}
\usepackage{textcomp}
\usepackage{xcolor}
\usepackage{mathtools}
\usepackage{multirow}
\usepackage{braket}
\usepackage{quantikz}
\usepackage{tikz}
\usepackage{graphicx}
\usepackage{totpages}
\usepackage{cleveref}
\usepackage{cancel}
\usepackage[normalem]{ulem}
\usepackage{siunitx}
\renewcommand{\selectlanguage}[1]{} 
\renewcommand{\bar}[1]{\overline{#1}}

\newcommand{\supp}{\mathrm{supp}}
\renewcommand{\proj}{\Pi}

\newcommand{\cP}{\mathcal{P}}
\newcommand{\cI}{\mathcal{I}}
\newcommand{\cJ}{\mathcal{J}}

\newtheorem{theorem}{Theorem}
\newtheorem{definition}{Definition}

\AddToHook{cmd/appendix/before}{%
  \crefalias{section}{appendix}%
  \crefalias{subsection}{subappendix}%
  \crefalias{subsubsection}{subsubappendix}%
}

\Crefname{equation}{Eq.}{Eqs.}
\crefname{equation}{Eq.}{Eqs.}
\Crefname{definition}{Def.}{Defs.}
\crefname{definition}{Def.}{Defs.}
\Crefname{figure}{Fig.}{Figs.}
\crefname{figure}{Fig.}{Figs.}
\Crefname{section}{Sec.}{Secs.}
\crefname{section}{Sec.}{Secs.}
\Crefname{assumption}{Assum.}{Assums.}
\crefname{assumption}{Assum.}{Assums.}
\Crefname{appendix}{Appendix}{Appendices}
\crefname{appendix}{Appendix}{Appendices}

\begin{document}

\preprint{APS/123-QED}

\title{In-Situ Simultaneous Magic State Injection on Arbitrary CSS qLDPC Codes}

\author{Kun Liu}
\email{kun.liu.kl944@yale.edu}
\affiliation{Department of Computer Science, Yale University, New Haven, Connecticut 06511, USA}
\affiliation{Yale Quantum Institute, Yale University, New Haven, Connecticut 06511, USA}

\author{Shifan Xu}
\affiliation{Yale Quantum Institute, Yale University, New Haven, Connecticut 06511, USA}
\affiliation{Department of Applied Physics, Yale University, New Haven, Connecticut 06511, USA}

\author{Tomas Jochym-O'Connor}
\affiliation{IBM Quantum, IBM T.J. Watson Research Center, Yorktown Heights, New York 10598, USA}

\author{Zhiyang He}
\affiliation{Department of Mathematics, Massachusetts Institute of Technology, Cambridge, MA 02139, USA}

\author{Shraddha Singh}
\affiliation{IBM Quantum, IBM T.J. Watson Research Center, Yorktown Heights, New York 10598, USA}

\author{Yongshan Ding}
\email{yongshan.ding@yale.edu}
\affiliation{Department of Computer Science, Yale University, New Haven, Connecticut 06511, USA}
\affiliation{Yale Quantum Institute, Yale University, New Haven, Connecticut 06511, USA}
\affiliation{Department of Applied Physics, Yale University, New Haven, Connecticut 06511, USA}

\date{\today}

\begin{abstract}

Quantum low-density parity-check (qLDPC) codes can encode many logical qubits within a single code block at low physical qubit overhead, yet magic state injection into such codes remains largely underexplored. Existing state injection proposals for qLDPC codes predominantly follow an external prepare-and-transfer paradigm, in which raw magic states are prepared outside the target code block and subsequently injected via inter-code operations. 
We propose the first \emph{in-situ} magic state injection: a scheme in which logical magic states are directly prepared within a qLDPC memory block, only using resources required for syndrome extraction.
We show that our scheme is generalizable to any CSS qLDPC code,
with examples of circuit-level simulations on the $[[144,12,12]]$ Bivariate Bicycle (BB) code and the $[[225,9,4]]$ Hypergraph Product code. 
We focus on a regime where correlated injection errors are negligible.
In the BB code, this corresponds to a configuration that simultaneously injects four logical $\ket{Y}$ states.
Under a uniform depolarizing noise model with physical error rate $10^{-3}$, this achieves an injection error rate of $1.62 \times 10^{-3}$ per logical qubit, while the correlated-error contribution is only $2 \times 10^{-5}$ per logical qubit (about $1\%$ of the injection error rate).
Under a hardware-motivated asymmetric noise model where single-qubit gate errors are $10\%$ of two-qubit gate errors,
the injection error rate per logical qubit falls to $ 6.7 \times 10^{-4} $, below the error rate ($ 10^{-3} $) of the two-qubit gates used to encode the magic states.
Its simplicity allows our scheme to be applied to arbitrary CSS qLDPC codes using only the ancilla qubits native to syndrome extraction, and yield a reduction in space overhead relative to both prepare-and-transfer approaches and surface-code-based magic state injection schemes.
\end{abstract}

\maketitle


\section{Introduction}
The Eastin-Knill theorem~\cite{eastin_restrictions_2009} implies that no quantum error-correcting code admits a fully transversal universal gate set. In particular, for many leading 2D code families, including surface codes and 2D color codes, logical non-Clifford gates are not available transversally~\cite{bravyi_classification_2013}. A standard route to universality is therefore to consume high-fidelity magic states through gate teleportation. The first step in such a pipeline is to prepare a noisy encoded magic state in the target codespace, a primitive commonly referred to as \emph{magic state injection} (MSI)~\cite{li_magic_2015,lao_magic_2022,gidney_cleaner_2023}. MSI for surface and color codes has been extensively developed and experimentally demonstrated~\cite{daguerre_experimental_2025, dasu_breaking_2025, kim_magic_2024,rodriguez2024experimental}, and can produce encoded magic states with logical error rates on the order of the error rate of the underlying physical operations. These injected states are then further purified by magic state distillation~\cite{bravyi_magic_2012, bravyi_universal_2005,rodriguez2024experimental,xu_distilling_2026}, cultivation~\cite{chamberland_very_2020, chen_efficient_2025-1, gidney_magic_2024, sahay_fold-transversal_2025}, or related protocols.

With the advancement of quantum low-density parity-check (qLDPC) codes~\cite{panteleev_asymptotically_2022, 
panteleev_degenerate_2021, tillich_quantum_2014, kovalev_quantum_2013,breuckmann_quantum_2021,gottesman_fault-tolerant_2014}, methods for parallel magic state generation have received recent interest. 
Unlike topological codes with large space overheads, qLDPC codes can encode many logical qubits in a single block at constant rate asymptotically. 
This raises a natural question: can one prepare \emph{multiple} logical non-Clifford resource states directly inside a single qLDPC memory block, thereby reducing the space overhead of magic state preparation? Answering this question requires an MSI primitive that is native to qLDPC memories rather than one that first prepares resource states elsewhere and then transfers them into the code block.

Existing approaches do not yet provide such a primitive. 
Recent proposals for qLDPC magic state preparation largely follow a \emph{prepare-and-transfer} paradigm, where magic states are first prepared in physical qubits or auxiliary logical blocks and then teleported into the target qLDPC memory~\cite{zhang_constant-overhead_2025,xu_batched_2025}. Related work has also explored implementing non-Clifford operations through code switching to specially structured codes, such as 3D lifted-product codes with transversal non-Clifford gates~\cite{li_transversal_2025}. These approaches are important advances, but they either require additional auxiliary structures or rely on specialized code families. What is still missing is a general \emph{in-situ} injection primitive that works for generic qLDPC memories.

In this work, we present a magic state injection scheme for arbitrary CSS codes. 
For a CSS qLDPC code with parameters $[[n,k,d]]$, our scheme can inject up to $k$ logical magic states simultaneously into a single code block. The construction uses only the ancilla resources already needed for syndrome extraction, and does not require separate auxiliary code blocks for state transfer. 
At a conceptual level, our result shows that MSI on multiple logical qubits can be carried out natively inside a qLDPC memory. Our work opens a route to non-Clifford resource preparation with lower space overhead in architectures where qLDPC memories are the main computational substrate.

To benchmark our MSI scheme, we conduct a case study on the $[[144,12,12]]$ bivariate bicycle (BB) code~\cite{bravyi_high-threshold_2024}.
Under a uniform depolarizing noise model with physical error rate $p=10^{-3}$, the generic construction yields a logical error rate per logical qubit, or the \emph{injection error rate}, of approximately $4.2 \times 10^{-2}$ when injecting all $k=12$ logical qubits.
This result demonstrates functionality, but the error rate is too high.
To improve performance, we reduce the number of injected magic states and introduce deterministic stabilizers which detect injection errors.
With this optimized scheme, under uniform depolarizing noise at physical error rate $10^{-3}$, we inject 4 logical qubits into the $[[144,12,12]]$ BB code with per logical qubit injection error rate as low as $1.62\times 10^{-3}$ and a discard rate of 93\%. 
We focus on this regime because it yields negligible correlated injection errors, contributing about $ 1\% $ of the injection error rate.
This is important for downstream magic state distillation schemes that assume negligible correlated noise on the input qubits.
Under a hardware-motivated asymmetric noise model where single-qubit gate errors are $10\%$ of two-qubit gate errors, the per logical qubit injection error rate further drops to $6.7\times 10^{-4}$, which is below the two-qubit gate error rate of $10^{-3}$. In this regime, our construction recovers the characteristic feature of Li's surface code injection scheme~\cite{li_magic_2015}: the injected magic state can have superior fidelity to the worst physical operations used to encode the magic states into the target codespace. We further apply the construction to the $[[225,9,4]]$ hypergraph product (HGP) code~\cite{xu_constant-overhead_2023}, where we inject up to 6 logical qubits and obtain an injection error rate per logical qubit of $8 \times 10^{-3}$ under uniform depolarizing noise.

Our current injection error rates do not yet match those achieved by the best surface code magic state injection schemes; nevertheless, this approach offers a meaningful advantage in space overhead at comparable code distance, as it prepares multiple logical magic states directly within a single qLDPC memory block — a compactness that makes it particularly attractive for near-term demonstrations and early fault-tolerant architectures where qubit footprint remains a central constraint. More broadly, the primary contribution of this work is to establish a qLDPC-native injection primitive that can serve as a natural foundation for future qLDPC-based cultivation and distillation protocols, and we anticipate that further optimization of logical Pauli representatives~\cite{cross_improved_2024,yoder_tour_2025} and syndrome extraction circuits~\cite{gidney_cleaner_2023,lao_magic_2022,singh_high-fidelity_2022} could yield material improvements in performance.

The remainder of this paper is organized as follows.
In \Cref{sec:k-injection}, we introduce the general $k$-injection scheme.
In \Cref{sec:improved-injection-scheme}, we present the improved injection scheme.
In \Cref{sec:evaluation}, we report the circuit-level simulation results.

\section{Preliminaries}
\label{sec:background}


\subsection{Magic state injection}

The Eastin-Knill theorem~\cite{eastin_restrictions_2009} states that no quantum error-correcting code admits a universal transversal gate set.
A standard route to universality is therefore to prepare high-fidelity magic states and consume them through gate teleportation or related protocols~\cite{litinski_game_2019,litinski_magic_2019}.
In this work, we focus on single-qubit magic states on the $XY$ plane of the Bloch sphere, namely the $+1$ eigenstates of observables of the form $\cos\theta\,X+\sin\theta\,Y$.
For $\theta=\pi/4$, this is the $H_{\mathrm{XY}}$ eigenstate, i.e., the $\ket{T}$ state, which is the target output of recent cultivation protocols~\cite{gidney_magic_2024,sahay_fold-transversal_2025,chen_efficient_2025-1}.

For a single-qubit rotation $R_z(\theta):=e^{-i\theta Z/2}$, define
\begin{equation}
|\theta\rangle := R_z(\theta)|+\rangle.
\end{equation}
This state is the $+1$ eigenstate of the \emph{rotated Pauli observable}~\cite{li_magic_2015}
\begin{equation}
M \, := \, R_z(\theta) X R_z(-\theta)
= \cos\theta\,X + \sin\theta\,Y.
\label{eq:single-qubit-magic-op}
\end{equation}
Since our formalism applies to arbitrary rotation angles, we suppress the angle $\theta$ in the notation whenever it is not essential.
When $M$ is supported on physical qubit $q$, we write it as $M_q$.

An $[[n,k,d]]$ stabilizer code is specified by a stabilizer group $S \subseteq \mathcal{P}_n$ with a chosen set of logical Pauli representatives $\{\bar X_i,\bar Z_i\}_{i=1}^k \subseteq N(S)\setminus S$, where $N(S)$ is the normalizer of $S$.
In this work we focus on CSS codes, so the stabilizer generators are taken to be either $X$-type or $Z$-type.
We write $\mathcal{S}=\{s_a\}_{a=1}^{n-k}$ for a chosen independent generating set, so that $S=\langle \mathcal{S}\rangle$.

For logical qubit $i\in [k]:=\{1,2,\dots,k\}$, we define the \emph{logical rotated Pauli observable}
\begin{equation}
\begin{aligned}
\bar M_i
&:= e^{-i\theta_i \bar Z_i/2}\,\bar X_i\,e^{+i\theta_i \bar Z_i/2} \\
&= \cos\theta_i\,\bar X_i + \sin\theta_i\,\bar Y_i,
\label{eq:logical-rotated-x}
\end{aligned}
\end{equation}
where $\bar Y_i:= i\bar X_i\bar Z_i$.
The corresponding logical magic state $|\bar\theta_i\rangle$ satisfies $\bar M_i|\bar\theta_i\rangle = |\bar\theta_i\rangle$.
The goal of magic state injection is therefore to prepare, with high fidelity and low qubit overhead, a code state that is a $+1$ eigenstate of $\bar M_i$ for one or more logical qubits in the same block.

In many surface code MSI schemes~\cite{li_magic_2015,lao_magic_2022,gidney_cleaner_2023}, state preparation is followed by rounds of syndrome extraction (SE) that detect faults and, up to a Pauli frame, project the data back into the codespace.
Pauli frame correction is optional and can be implemented in software.
Abstractly, one round of SE can be modeled as the codespace projector
\begin{equation}
\proj_\mathcal{S} \,=\, \prod_{s\in\mathcal{S}} \frac{I+s}{2}.
\label{eq:projector}
\end{equation}
Because $\bar X_i$, $\bar Y_i$, and $\bar Z_i$ all lie in $N(S)$, they commute with every stabilizer in $\mathcal{S}$, and hence
\begin{equation}
[\bar M_i,\,\proj_\mathcal{S}] = 0.
\label{eq:commute-with-projector}
\end{equation}
Therefore, if the prepared state is a $+1$ eigenstate of desired logical rotated Pauli observables $\set{\bar M_i}$, then projection into the codespace preserves that eigenvalue.
This observation is the starting point of our in-situ injection scheme: the design problem reduces to characterizing product-state initializations whose projection yields the target logical rotated observables.

\subsection{High-rate qLDPC codes}

Our main examples are the $[[144,12,12]]$ bivariate bicycle (BB) code~\cite{bravyi_high-threshold_2024} and the $[[225,9,4]]$ hypergraph product (HGP) code~\cite{xu_constant-overhead_2023}.
These are two canonical CSS qLDPC code families in the recent literature, and both encode multiple logical qubits in a single block, making them natural testbeds for simultaneous magic state injection on multiple logical qubits.
They also admit explicit syndrome extraction circuits, which is essential because the performance of our scheme depends not only on the code itself, but also on the chosen logical Pauli representatives and on the syndrome extraction schedule.
We defer the detailed constructions of the logical representatives and syndrome extraction circuits used in this work to \cref{app:high-rate-qLDPC-codes}.

\section{General k-injection}
\label{sec:k-injection} 
In this section, we present our general $ k $-qubit injection scheme for arbitrary codes. The first step is to
find a set of conjugate logical Pauli representatives $ \set{\bar{X}_i', \bar{Z}_i'} $ such that the support of all of the representatives can be partitioned into two sets, $\mathcal{Q}_L$ and $\mathcal{Q}_P$, where $|\mathcal{Q}_L| = k$. 
Moreover, for every qubit $j$ in $\mathcal{Q}_P$, there is a fixed single-qubit Pauli operator $P_j$ such that
\begin{align}
   \forall i,  (\bar{X}_i')_j, (\bar{Z}_i')_j \in \{ I, P_j \}.
\end{align}
In other words, the representatives \textit{locally commute} on qubits in $\mathcal{Q}_P$.
Given such logical representatives, one can prepare an arbitrary target state on the $k$~qubits in~$\mathcal{Q}_L$ and prepare the qubits in $\mathcal{Q}_P$ in single-qubit states stabilized by $\{ P_j \}_{j \in \mathcal{Q}_P}$.
Consequently, the logical representatives are prepared in the target state.
One would then project into the codespace by measuring all of the stabilizers of the code which, in the absence of noise, will not disturb the prepared logical representatives. 

We describe how to find such a basis using a stabilizer-cleaning argument in \Cref{app:k-injection-details}, however there are equivalent approaches that achieve the same result by considering the binary symplectic matrices formed by the stabilizer generators and logical operators~\cite{gottesman_stabilizer_1997,cowtan_parallel_2025,shi_stabilizer_2025}.
We explain in \cref{app:k-injection-details} why the stabilizer-cleaning preserves the encoded logical information. 

The stabilizer-cleaning produces a new set of logical representatives $ \mathcal{L}' := \{ \bar{X}_i', \bar{Z}_i' \}_{i=1}^{k} $ together with an updated stabilizer generator set $ \mathcal{S}' $. The representatives in $ \mathcal{L}' $ are supported on $ k $ \emph{carrier qubits}, denoted by $\mathcal{Q}_L = \set{Q_i}_{i=1}^k$, while the remaining \emph{peeled qubits}, denoted by $\mathcal{Q}_P$, are fixed by the weight-$1$ Pauli generators in $ \mathcal{S}' $.
Since $ \mathcal{L}' $ is again a valid conjugate logical Pauli basis, there exists a Clifford circuit $C$ on the carrier qubits that maps the standard single-qubit Paulis on $\mathcal{Q}_L$ to the reduced logical representatives:
\begin{equation}
C X_{Q_i} C^\dagger = \bar{X}_i',
\qquad
C Z_{Q_i} C^\dagger = \bar{Z}_i'.
\label{eq:tomas-step-clifford-circuit}
\end{equation}

In summary, the whole $ k $-qubit injection scheme is:
\begin{enumerate}
    \item Prepare the $ k $ carrier qubits $\set{Q_i}$ in the $+1$ eigenstates of $ M_{Q_i} $ for all $ i \in [k] $.
    \item Apply the Clifford circuit $C$ on the carrier qubits. (Such Clifford could be free. See discussion in \Cref{app:k-injection-details}).
    \item Prepare the $ n-k $ peeled qubits in the eigenstates of the single-qubit Paulis in $ \mathcal{S}' $.
    \item Measure the stabilizers of the original code, i.e., $ \mathcal{S} $.
\end{enumerate}

In this way, we inject $ k $ logical magic states encoded by $ \mathcal{L} $ simultaneously.

\subsection{Analysis of injection error rate}

Magic state injection schemes rely on encoding a set of qubits into the target codespace via syndrome extraction. This renders all stabilizer outcomes non-deterministic in the first SE round, providing no reliable check for post-selection or for suppressing preparation faults. While preparation faults on the magic-state qubits can never be detected, a careful choice of initial states can allow fixed stabilizers to detect faults on other qubits initialized in the $\ket{0}$ or $\ket{+}$ states. A general construction aimed at injecting the maximum number of logical qubits ignores this possibility. Consequently, faults occurring during state preparation and the first SE round can contribute directly to logical injection errors.

We observe this behavior explicitly in circuit-level simulation (assuming the Clifford $C$ is free). For the \( [[144,12,12]] \) BB code, when injecting all \(k=12\) logical qubits under the uniform depolarizing noise model of \Cref{sec:noise-model} with physical error rate \(p=10^{-3}\), the injection error rate (per logical qubit) is \(0.0422\), which is \(42.2\times\) of the physical error rate. 
This shows that utilizing first-round stabilizers to reduce the injection error rate might be necessary.

To this end, in the next section, we will introduce an improved injection scheme to reduce the injection error rate.

\section{Noise-aware injection scheme}
\label{sec:improved-injection-scheme}

In this section, we modify the scheme in \Cref{sec:k-injection} to improve the injection error rate via post-selection, inspired by Ref.~\cite{li_magic_2015}. The key idea is to directly prepare the state in the $+1$ eigenstate of the desired logical observables $ \bar M_i $,
under the \emph{original} logical representatives $\mathcal{L}$.
This method has two advantages: it removes the use of Clifford rotation, and enables post-selection using stabilizers with deterministic outcome in the first SE round. We call these \emph{fixed stabilizers}. Once an error flips the measurement outcome of any fixed stabilizer, we discard the shot and thus decrease the logical error rate of the kept shots.
Such fixed stabilizers do not exist in the general $k$-injection scheme since all the stabilizers in $\mathcal{S}$ anticommute with the prepared single-qubit Pauli states by construction.

An example of injecting two logicals $i$ and $j$ using the improved injection scheme is shown in \Cref{fig:inject-two-logicals}(d).
In the following, we first introduce sufficient conditions to prepare the target state.
Then, we introduce how to post-select on the measurement outcomes of fixed stabilizers to improve the injection error rate.

\begin{figure*}[htb]
    \centering
    \includegraphics[width=\textwidth]{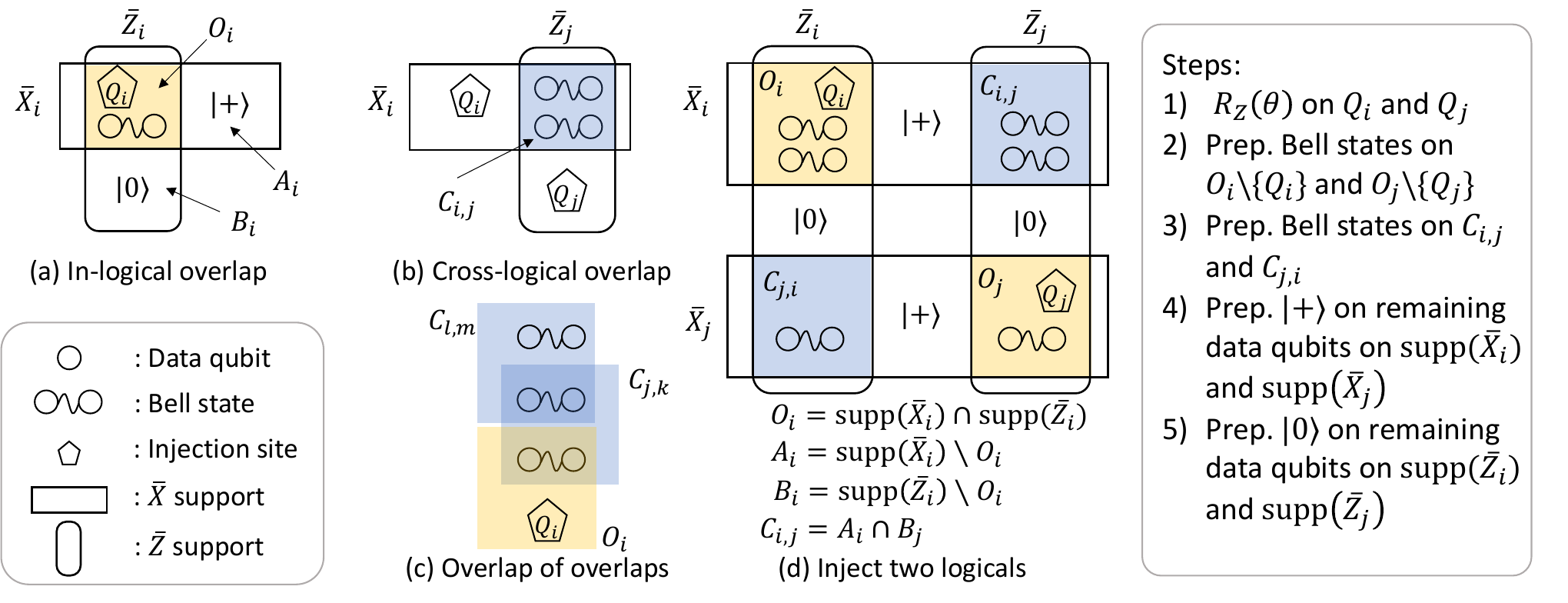}
    \caption{High-level summary of this work.
    (a) In-logical overlap: Typical injection scheme for a single logical qubit.
    For surface code, we perform $ R_z(\theta) $ on the injection site $ Q_i $,
    which is the single-qubit sitting in the overlap of its logical $ X $ and $ Z $ supports. But for qLDPC codes, $ O_{i} := \supp(\bar X_i) \cap \supp(\bar Z_i) $ has cardinality above 1.
    The solution is to prepare Bell states on $ O_{i} \setminus \{Q_i\} $.
    (b) Cross-logical overlap (remainder overlap): When injecting two logical qubits, $ i $ and $ j $, the overlap of their supports, i.e., $ C_{i,j} $, may not be empty.
    The solution is to prepare Bell states on $ C_{i,j} $.
    (c) Overlap of overlaps: the above-mentioned types of overlaps will overlap with each other.
    The solution is to identify a \emph{compatible pairing} (\Cref{def:generic-compatible}) s.t. the overlaps can be filled with Bell states.
    (d) An example of injecting two logical qubits, $ i $ and $ j $.
    Prepare Bell states on the overlaps,
    and initialize the rest of the qubits in the appropriate Pauli eigenbases.
    }
    \label{fig:inject-two-logicals}
\end{figure*}

The challenge in preparing the desired eigenstate of $ \bar M_i $'s is to deal with the overlaps of the logical $ X $ and $ Z $ supports due to the nature of qLDPC codes.
The central idea is to resolve these overlaps by Bell state preparation on carefully chosen qubit pairs.

We formalize as follows three overlap types, in-logical overlap, cross-logical overlap, and overlap of overlaps, as shown in \Cref{fig:inject-two-logicals}(a,b,c).
For each injected logical index $i\in\cI\subseteq [k]$, define the in-logical overlap $O_i$
as the overlap of the logical $ X $ and $ Z $ supports:
\begin{equation}
O_i \;:=\; \supp(\bar X_i)\cap\supp(\bar Z_i).
\end{equation}
The \emph{$X$-remainder} and \emph{$Z$-remainder} are defined as the qubits outside the in-logical overlap:
\begin{equation}
A_i \;:=\; \supp(\bar X_i)\setminus O_i, \quad
B_i \;:=\; \supp(\bar Z_i)\setminus O_i.
\end{equation}
Because $\{\bar X_i,\bar Z_i\}=0$ in CSS codes, $|O_i|$ is odd.
Choose an \emph{injection site} $Q_i\in O_i$, where the qubit $Q_i$ is prepared in the $+1$ eigenstate of $M_{Q_i}$. The remaining qubits in the in-logical overlap are denoted as $O_i\setminus\{Q_i\}$.
To treat in-logical and cross-logical overlaps uniformly, we define
\begin{equation}
    C_{i,j} =
    \begin{cases}
        A_i\cap B_j, & i\neq j,\\
        O_i\setminus\{Q_i\}, & i=j.
    \end{cases}
    \label{eq:unified-Cij}
\end{equation}
and
\begin{equation}
C := \bigcup_{i,j\in\cI} C_{i,j}.
\end{equation}

The goal of preparing
the state in the $+1$ eigenstate of the desired logical rotated Pauli observables $ \set{\bar M_i} $ is to satisfy the following logical constraints:
\begin{equation}
\bar M_i|\psi_{\mathrm{in}}\rangle = |\psi_{\mathrm{in}}\rangle,\quad \forall i\in\cI.
\label{eq:goal-pre-se-state}
\end{equation}

Since $[\bar M_i,\proj_{\mathcal S}]=0$, projection by ideal SE preserves the logical constraints:
\begin{equation}
\bar M_i|\psi_{\mathrm{out}}\rangle = |\psi_{\mathrm{out}}\rangle,\quad \forall i\in\cI.
\end{equation}

The goal in \Cref{eq:goal-pre-se-state} holds
if the following two \emph{sufficient} conditions hold:

\begin{definition}[Injection site independence]
    \label{def:injection-site-independence}
\begin{equation}
Q_i \notin \supp(\bar X_j)\cup \supp(\bar Z_j) \qquad \forall i\neq j \text{ in }\cI.
\end{equation}
\end{definition}

\begin{definition}[Compatible pairing]
    \label{def:generic-compatible}
    A pairing $\cP=\{\{u_t,v_t\}\}_{t=1}^T$ of $C$ is \emph{compatible} if for every $i,j\in\cI$ and every pair $\{u,v\}\in\cP$,
    \begin{equation}
    u\in C_{i,j}\ \Longleftrightarrow\ v\in C_{i,j}.
    \end{equation}
    Equivalently, each pair is either fully inside or disjoint from every $C_{i,j}$.
\end{definition}

\begin{theorem}[Sufficient conditions for improved injection scheme]
\label{thm:unified-sufficient-condition}
Assume injection site independence (\Cref{def:injection-site-independence}) and a compatible pairing (\Cref{def:generic-compatible}).
Prepare the state by:
\begin{enumerate}
    \item Bell state on every pair in $\cP$;
    \item $\ket{+}$ on $A_i\setminus C$ for each $i\in\cI$;
    \item $\ket{0}$ on $B_i\setminus C$ for each $i\in\cI$;
    \item the $+1$ eigenstate of $M_{Q_i}$ on each $Q_i$.
\end{enumerate}
Then \Cref{eq:goal-pre-se-state} holds, i.e.,
$
\bar M_i|\psi_{\mathrm{in}}\rangle = |\psi_{\mathrm{in}}\rangle
$, $\forall i\in\cI$.
\end{theorem}
The proof is given in \Cref{sec:proof-all-in-one}. To satisfy injection site independence and compatible pairing,
we might not be able to inject the full set of $ k $ logicals.
We discuss how to find maximal injectable sets in \Cref{sec:maximal-injectable-sets}.


After fixing the preparation on qubits in
$\supp(\bar X_i)\cup \supp(\bar Z_i)$ for all injected logicals,
the remaining design freedom lies in how we initialize qubits outside these supports.
We use this freedom to increase the protection of the injected logical observables in the first syndrome extraction round.

The key mechanism is post-selection on \emph{fixed stabilizers}.
Our initialization pattern ensures that in the absence of any noise, all the fixed stabilizers deterministically yield a $+1$ outcome.
In the presence of noise, if any fixed stabilizer is measured as $-1$, we discard the shot.
Thus, the usefulness of a given initialization pattern is determined by which first-order error mechanisms can be detected through flipped fixed stabilizers.

Our objective is to choose the remaining qubit initializations so as to maximize the fixed-stabilizer protection of the injected logical supports.
To keep the search space manageable and to avoid introducing additional first-round complexities, we restrict these extra choices to single-qubit Pauli-basis preparations.

For qLDPC codes, this optimization problem is substantially more challenging than in surface code injection schemes.
In \Cref{sec:max-protection-injected-supports}, we formulate the search for a high-protection initialization pattern as a mixed-integer linear programming (MILP) problem and solve it with off-the-shelf solvers.

\section{Evaluation}
\label{sec:evaluation}

In this section, we evaluate the improved injection scheme under circuit-level noise. Our main goals are to validate the first-order analytical logical noise model (see \cref{sec:first-order}) against Monte Carlo simulation, to characterize how the injection error scales with the number of injected logical qubits, and to assess the sensitivity of the results to both the noise model and the optimization objective. We use the $[[144,12,12]]$ BB code as the main testbed and report four metrics: the total post-selected injection error rate, the correlated injection error rate, the discard rate, and the approximate injection error rate per logical qubit, when correlated errors are negligible. We then use the $[[225,9,4]]$ HGP code to test the generality of the proposed method, and compare against the best known surface code injection scheme to place our current results in context.

\subsection{Injection procedure}

We evaluate the improved injection scheme using the following simulation protocol:
\begin{enumerate}
    \item Prepare the state as in \Cref{thm:unified-sufficient-condition}.
    \item Perform a round of SE.
    \item Post-select on the fixed stabilizers. If the first-round measurement outcomes of the fixed stabilizers are not all $+1$, we discard the shot and restart the injection procedure. The corresponding detectors are called \emph{fixed-stabilizer detectors}.
    \item Perform a second round of SE. We compare the full stabilizer outcome of the second round with that of the first round, and discard the shot if they disagree. The corresponding detectors are called \emph{round-parity detectors}.
    \item Perform a further $r_2$ rounds of SE and decode using all retained syndrome information.
\end{enumerate}

This protocol is chosen because the dominant contribution to the injection error arises from faults before and during the first SE round. 
Additional SE rounds can further suppress higher-order contributions, while potentially introducing more memory faults. 
As a point of clarity, all stabilizers are measured in Step 2, but the shot is discarded at this stage only if the subset of fixed stabilizers do not give a $+1$ outcome. The discarding in Step 4 is a result of comparing \textit{all} stabilizer measurements with their corresponding result in Step 2.
In our simulations, taking one extra decoding round beyond the two post-selection rounds captures the relevant behavior without materially changing the leading-order injection error rate. Accordingly, throughout this section we set $r_2 = 1$ or $ 2 $ for the sake of simulation efficiency. We use \texttt{stim} to simulate. Since \texttt{stim} only supports Clifford operations, we follow~\cite{gidney_cleaner_2023,sahay_fold-transversal_2025} and inject $\ket{\bar{Y}}$ states.

\subsection{Noise model}
\label{sec:noise-model}

We employ the depolarizing noise model used in prior analyses of surface code injections~\cite{li_magic_2015,lao_magic_2022} for our simulations as summarized in \Cref{tab:li_noise_channels,tab:li_noise_model}. A qubit initialized in a $+1$ Pauli eigenstate may be flipped to the corresponding $-1$ eigenstate with probability $p_{\rm IN}$. A measurement outcome may be flipped with probability $p_{\rm M}$. An idling qubit undergoes a single-qubit depolarizing channel with error probability $p_{\rm IDLE}$. A single-qubit gate is followed by a single-qubit depolarizing channel with probability $p_{\rm 1Q}$, i.e., one of $\{X,Y,Z\}$ is applied uniformly at random with probability $p_{\rm 1Q}/3$. A two-qubit gate is followed by a two-qubit depolarizing channel with error probability $p_{\rm 2Q}$, i.e., one of the 15 nontrivial Pauli errors in $\{I,X,Y,Z\}^{\otimes 2} \setminus \{II\}$ is applied with probability $p_{\rm 2Q}/15$.

We consider two concrete parameter settings. The first is the \emph{uniform depolarizing noise model}, where $p_{\rm IDLE} = p_{\rm IN} = p_{\rm M} = p_{\rm 1Q} = p_{\rm 2Q} = p$. The second is the \emph{hardware-motivated asymmetric noise model}, where error rates of the single-qubit gates are 10\% of the two-qubit gates, i.e., $p_{\rm IDLE} = p_{\rm IN} = p_{\rm M} = p_{\rm 1Q} = 0.1p_{\rm 2Q} = 0.1p$. Unless otherwise stated, the results below use one of these two models.

We merge the state-preparation circuit and the first SE round whenever possible. This is important because idling errors before the first SE round can increase undetectable first-order initialization errors.

\subsection{First-order analytical results}
\label{sec:first-order}

We begin with a first-order analytical model for the BB code injection performance. Let $\cJ \subseteq \cI$ be a subset of the injected logical qubits, and let $p_\cJ$ denote the post-selected injection error rate when the incorrect logicals are \emph{exactly} those in $\cJ$:
\begin{equation}
p_\cJ := \Pr(\text{incorrect logicals are exactly } \cJ).
\label{eq:p-cJ}
\end{equation}

The \emph{total} injection error rate is then
\begin{equation}
    p_{\cI}^{\rm tot} := \sum_{\substack{\forall \cJ \subseteq \cI \\ \cJ \neq \emptyset}} p_\cJ,
    \label{eq:p-cI}
\end{equation}
and the \emph{correlated} injection error rate is
\begin{equation}
    p^{\text{corr}}_{\cI} := \sum_{\substack{\forall \cJ \subseteq \cI \\ |\cJ| > 1}} p_\cJ.
    \label{eq:p-corr-cI}
\end{equation}
The first metric captures the overall probability of a logical injection failure after post-selection, while the second isolates the component involving multiple logical qubits, which has no analogue in single logical surface code injection.

Most injection errors occur before and during the first SE round. Following~\cite{li_magic_2015,lao_magic_2022}, we analyze the first-order contribution to the injection error rate.
In other words, we consider weight-one fault mechanisms that occur before or during the first SE round.

An \emph{injection configuration} is specified by the injected logical set $\cI$ and the corresponding injection sites $\{Q_i\}_{i\in\cI}$. An \emph{initial configuration} further specifies the initializations used for that injection configuration. We use the optimization procedure in \Cref{sec:max-protection-injected-supports}
to find, for each feasible injection configuration, the initial configuration that minimizes $p_{\cI}^{\rm tot}$. By enumerating the feasible configurations on the $[[144,12,12]]$ BB code, we find solutions for $|\cI|=1$ through 6.

\begin{table}[htb]
    \centering
    \resizebox{\columnwidth}{!}{
\begin{tabular}{ccccccccc}
    \toprule
    $|\mathcal{I}|$ & $\frac{p_{\cI}^{\rm tot}}{|\mathcal{I}|}$ & $p_{\cI}^{\rm tot}$ & $\frac{p^{\text{corr}}_\cI}{|\mathcal{I}|}$ & $p^{\text{corr}}_\cI$ & $\frac{p^{\text{corr}}_\cI}{p_{\cI}^{\rm tot}}$ & $|F|$ & $r_{\rm disc}^{\rm UP}(0.001)$ & $|\mathcal{P}|$ \\
    \midrule
    1 & 1.40$p$ & 1.40$p$ & 0 & 0 & 0 & 26 & 0.9298 & 0 \\
    2 & 1.43$p$ & 2.87$p$ & 0 & 0 & 0 & 32 & 0.9315 & 0 \\
    3 & 1.53$p$ & 4.60$p$ & 0 & 0 & 0 & 31 & 0.9315 & 2 \\
    4 & 1.57$p$ & 6.27$p$ & 0.02$p$ & 0.07$p$ & 0.01 & 32 & 0.9318 & 3 \\
    5 & 6.04$p$ & 30.20$p$ & 0.96$p$ & 4.80$p$ & 0.16 & 28 & 0.9311 & 5 \\
    6 & 13.76$p$ & 82.53$p$ & 1.77$p$ & 10.60$p$ & 0.13 & 18 & 0.9281 & 10 \\
    \bottomrule
\end{tabular}
}
\caption{
Under the uniform depolarizing noise model, first-order analytical results for the best BB code initial configuration found at each $|\mathcal{I}|$. Here $p_{\cI}^{\rm tot}$ is the total injection error rate defined in \Cref{eq:p-cI}; $p^{\text{corr}}_\cI$ is the correlated injection error rate defined in \Cref{eq:p-corr-cI}; $|F|$ is the number of fixed stabilizers in the first SE round; $r_{\rm disc}^{\rm UP}$ is the union proxy for the discard rate defined in \Cref{eq:rdisc_up_def}; and $|\mathcal{P}|$ is the number of Bell states. When correlated errors are negligible, $p_{\cI}^{\rm tot}/|\mathcal{I}|$ can be interpreted as an approximate injection error rate per logical qubit.
}
\label{tab:ana-error-rate}
\end{table}

The results are summarized in \Cref{tab:ana-error-rate}. For $|\cI|=1,2,3$, the best configurations have no first-order correlated errors. For $|\cI|=4$, the correlated component remains small: $p^{\rm corr}_{\cI}/p_{\cI}^{\rm tot}\approx 1\%$. In this regime, $p_{\cI}^{\rm tot}/|\cI|$ is a meaningful approximation to the injection error rate per logical qubit, and it lies between $1.40p$ and $1.57p$. As a reference point, Li's surface code injection scheme has first-order error rate $46/15\,p \approx 3.07p$ under the same uniform depolarizing model~\cite{li_magic_2015}. 
In contrast, once $|\cI|$ increases to 5 or 6, both the total injection error and the correlated error component rise sharply. For these larger-$|\cI|$ configurations, $p_{\cI}^{\rm tot}/|\cI|$ is best viewed as a summary statistic rather than a clean per logical qubit error rate.

We also derive the union proxy of the discard rate $r_{\rm disc}^{\rm UP}$ by applying union bounds on the detectable error mechanisms in the detector error model. Details are given in \Cref{sec:detector-error-model}.
Across the best configurations in \Cref{tab:ana-error-rate}, the discard proxy varies only weakly even though the number of fixed stabilizers changes noticeably. For example, from $|\cI|=5$ to $|\cI|=6$, the number of fixed stabilizers drops from 28 to 18, whereas $r_{\rm disc}^{\rm UP}$ changes only from 0.9311 to 0.9281. This indicates that the discard behavior is not determined solely by the number of fixed stabilizers. Instead, a substantial part of the discard probability comes from mechanisms that are detected only after comparing the first two SE rounds, i.e., by the round-parity detectors. We explain this behavior in \Cref{sec:explanation-similar-discard-rate}.

\subsection{Circuit-level simulation vs analytical results}

Next, we compare the first-order analysis with circuit-level Monte Carlo simulation. We choose the $|\cI|=4$ configuration from \Cref{tab:ana-error-rate}, namely $\cI=\{4,5,7,10\}$, because its correlated injection error rate remains negligible.
The resulting error breakdown is shown in \Cref{fig:LER-breakdown-4-491}. The numerical results agree closely with the analytical first-order prediction, both for the total injection error rate and for the distribution of the dominant $p_\cJ$ terms.
Contributions that are nearly invisible correspond to subsets $\cJ$ whose logical error rates are second order or higher in $p$.
Notably, the chosen configuration gives very little contribution to the correlated injection errors at first-order.

The error from initialization ($p_{\rm IN}$) only contributes to $p_{\mathcal{J}}$ where $|\mathcal{J}
| = 1$ (independent errors in \Cref{fig:LER-breakdown-4-491}) with $1\cdot p=0.001$. 
This reflects the structure of the optimized configuration: aside from the injection sites themselves, \emph{all} qubits on the injected logical supports are protected by fixed stabilizers in the first SE round.

\begin{figure}[htb]
    \centering
    \includegraphics[width=\columnwidth]{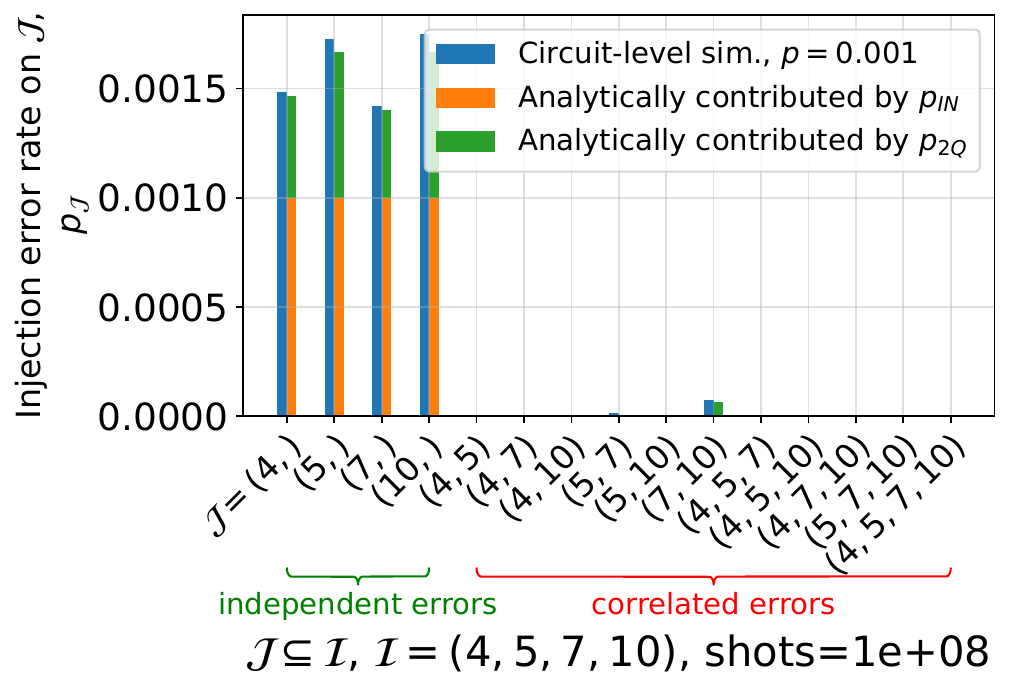}
    \caption{Breakdown of the circuit-level injection error rates for the $|\cI|=4$ BB code configuration $\cI=\{4,5,7,10\}$ under the uniform depolarizing noise model at $p=0.001$. The simulation uses $10^{8}$ shots and $r=2$ SE rounds. Each bar corresponds to a nonempty subset $\cJ\subseteq\cI$, where $p_\cJ$ is the post-selected injection error rate when exactly the logicals in $\cJ$ are incorrect. Summing all bars yields the total injection error rate $p_{\cI}^{\rm tot}=0.00647$, in close agreement with the analytical value $6.27p=0.00627$ from \Cref{tab:ana-error-rate}.}
    \label{fig:LER-breakdown-4-491}
\end{figure}

\subsection{Scalability analysis}

We now compare the first-order analysis and Monte Carlo simulation across all best BB code configurations in \Cref{tab:ana-error-rate}. The results in \Cref{fig:scalability} show close agreement for both the injection error rate and the discard rate. In particular, the low-correlation regime extends through $|\cI|=4$, where $p_{\cI}^{\rm tot}/|\cI|$ remains nearly flat, while the sharp increase at $|\cI|=5$ and 6 is reproduced by both analysis and simulation. This identifies $|\cI|=4$ as the largest practical BB code configurations we found under the uniform depolarizing model.

The empirical discard rate $r_{\rm disc}^{\rm MC}$ also tracks the union proxy $r_{\rm disc}^{\rm UP}$ closely. For clarity, \Cref{fig:scalability} only overlays the union proxy for the $|\cI|=4$ configuration, but the same agreement holds across the other best configurations as well.

\begin{figure}[htb]
    \centering
    \includegraphics[width=\columnwidth]{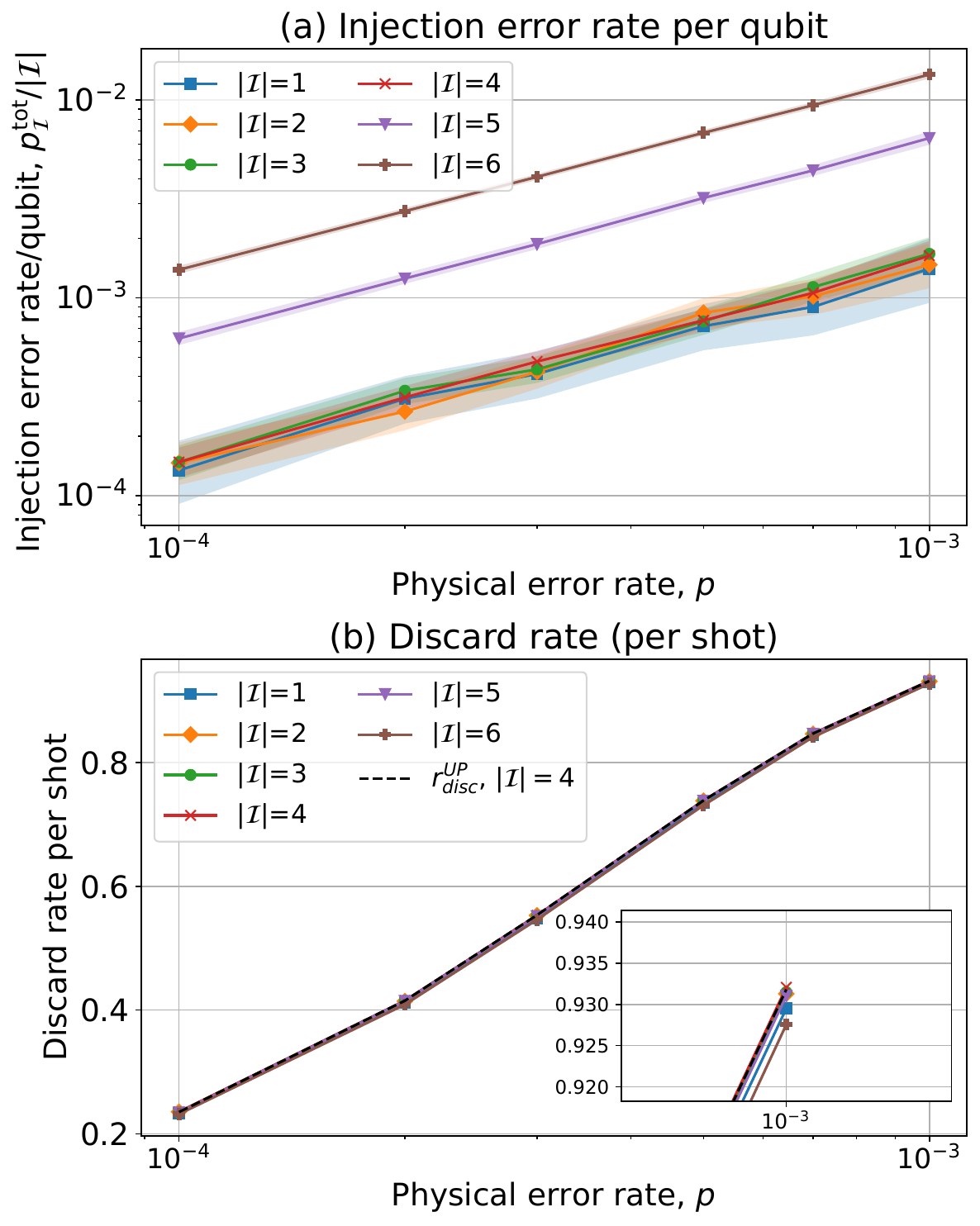}
    \caption{
    Circuit-level simulation results of $[[144,12,12]]$ BB code under the uniform depolarizing noise model.
    (a) Injection error rate per logical qubit and (b) discard rate for the best BB code configurations at each $|\cI|$ in \Cref{tab:ana-error-rate}. The union proxy $r_{\rm disc}^{\rm UP}$ matches the empirical discard rate $r_{\rm disc}^{\rm MC}$ closely; for clarity we overlay the proxy only for the $|\cI|=4$ configuration.
    }
    \label{fig:scalability}
\end{figure}

\subsection{Hardware-motivated asymmetric noise model}

We also re-optimize the injection and initial configurations under the hardware-motivated asymmetric noise model. For the $[[144,12,12]]$ BB code, the same best configurations remain optimal at each $|\cI|$ as in the uniform depolarizing model. The analytical results are listed in \Cref{tab:ana-error-rate-li}, and the corresponding circuit-level results are shown in \Cref{fig:scalability-li}.

\begin{table}[htb]
    \centering
    \resizebox{0.96\columnwidth}{!}{
\begin{tabular}{ccccccccc}
    \toprule
    $|\mathcal{I}|$ & $\frac{p_\cI^{\rm tot}}{|\mathcal{I}|}$ & $p_\cI^{\rm tot}$ & $\frac{p^{\text{corr}}_\cI}{|\mathcal{I}|}$ & $p^{\text{corr}}_\cI$ & $\frac{p^{\text{corr}}_\cI}{p_\cI^{\rm tot}}$ & $|F|$ & $r_{\rm disc}^{\rm UP}(0.001)$ & $|\mathcal{P}|$ \\
    \midrule
    1 & 0.50$p$ & 0.50$p$ & 0 & 0 & 0 & 26 & 0.8082 & 0 \\
    2 & 0.53$p$ & 1.07$p$ & 0 & 0 & 0 & 32 & 0.8100 & 0 \\
    3 & 0.63$p$ & 1.90$p$ & 0 & 0 & 0 & 31 & 0.8101 & 2 \\
    4 & 0.67$p$ & 2.67$p$ & 0.02$p$ & 0.07$p$ & 0.02 & 32 & 0.8104 & 3 \\
    5 & 2.55$p$ & 12.77$p$ & 0.51$p$ & 2.57$p$ & 0.20 & 28 & 0.8099 & 5 \\
    6 & 5.61$p$ & 33.63$p$ & 0.92$p$ & 5.50$p$ & 0.16 & 18 & 0.8070 & 10 \\
    \bottomrule
\end{tabular}
}
\caption{
Analytical results for the best BB code initial configuration found at each $|\mathcal{I}|$ under the $p_{\rm 1Q}=0.1p_{\rm 2Q}$ noise model. The notation is the same as in \Cref{tab:ana-error-rate}. When correlated errors are negligible, $p_{\cI}^{\rm tot}/|\mathcal{I}|$ can be interpreted as an approximate injection error rate per logical qubit.
}
\label{tab:ana-error-rate-li}
\end{table}

Compared with the uniform model, the leading coefficients are substantially smaller because initialization, measurement, idle, and single-qubit gate faults are all downweighted relative to the two-qubit gate faults. For $|\cI|\in\{1,2,3,4\}$, the approximate injection error rate drops to $0.50p$--$0.67p$, below the two-qubit gate error rate $p$. The same qualitative transition occurs at larger $|\cI|$: the low-correlation regime persists through $|\cI|=4$, while the error rate rises sharply for $|\cI|\in\{5,6\}$.

\begin{figure}[htb]
    \centering
    \includegraphics[width=\columnwidth]{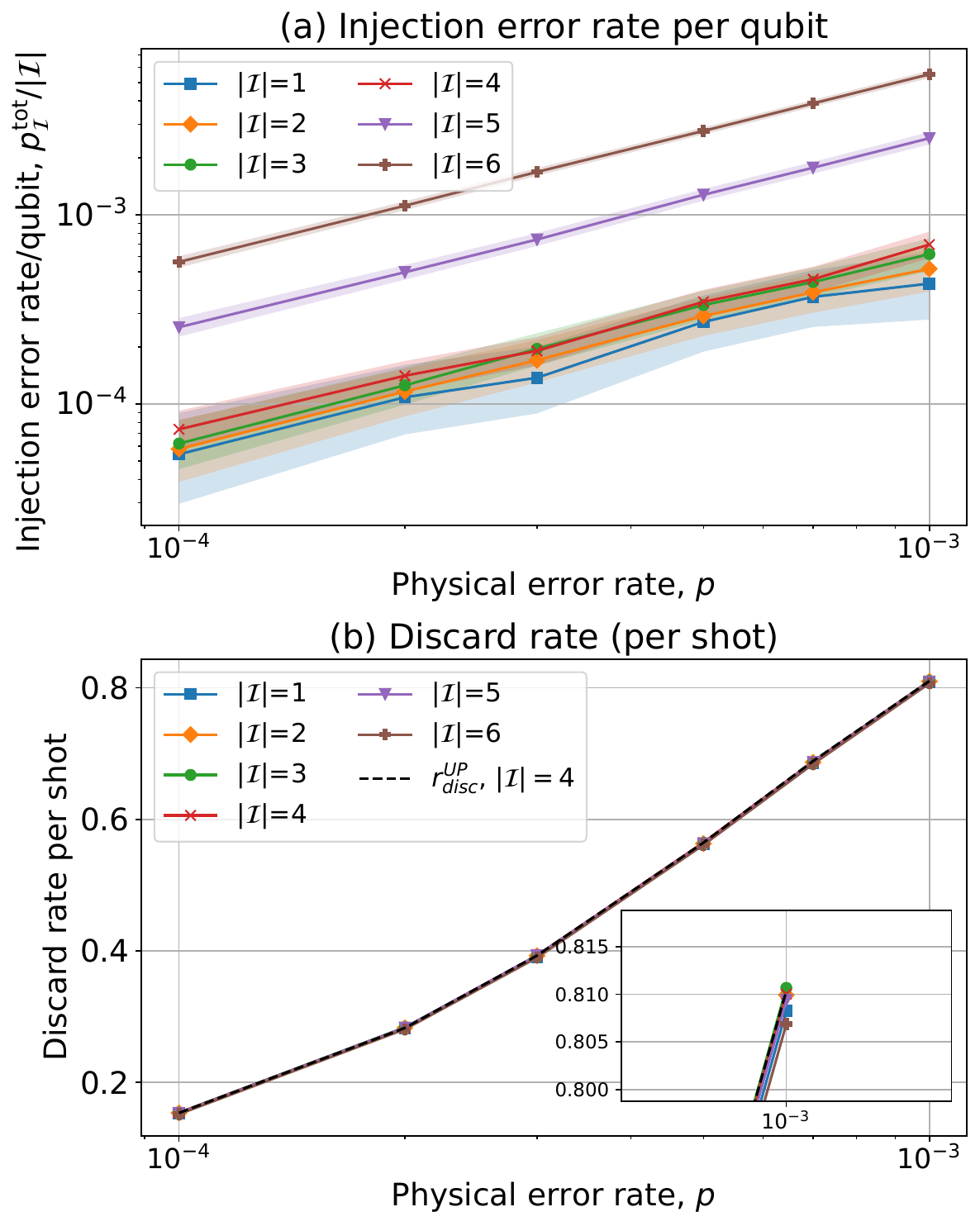}
    \caption{
        Circuit-level simulation results of $[[144,12,12]]$ BB code under the hardware-motivated asymmetric noise model.
        (a) Injection error rate per logical qubit and (b) discard rate for the best BB code configurations at each $|\cI|$ in \Cref{tab:ana-error-rate-li}.
    }
    \label{fig:scalability-li}
\end{figure}

\subsection{Comparison of optimization objectives}

\begin{figure}[htb]
    \centering
    \includegraphics[width=\columnwidth]{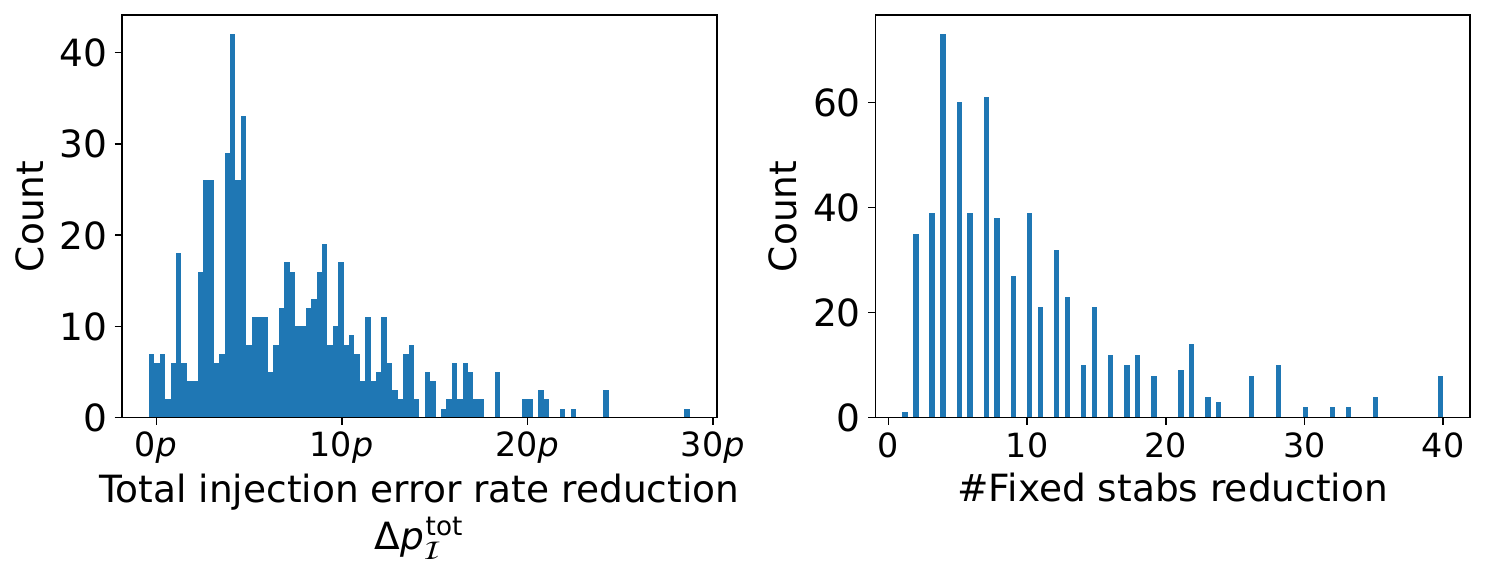}
    \caption{
        Histogram of the change in total injection error rate and in the number of fixed stabilizers when switching the objective from maximizing the total number of fixed stabilizers to maximizing the number of protected support qubits. All feasible BB code injection configurations are evaluated under the uniform depolarizing noise model.
    }
    \label{fig:cmp-fixing-methods}
\end{figure}

We next compare two optimization objectives for choosing the initial configuration: maximizing the total number of fixed stabilizers, and maximizing the number of protected qubits on the supports of the injected logical observables. For each feasible BB code injection configuration, we optimize once under each objective and compare the resulting first-order total injection error rates.

As shown in \Cref{fig:cmp-fixing-methods}, maximizing protected support qubits consistently yields lower injection error rates across all enumerated configurations, even though it often produces fewer fixed stabilizers overall. This confirms that first-order performance is governed more directly by whether faults on the injected logical supports are detected than by the raw count of fixed stabilizers. In other words, fixed stabilizers that act primarily on qubits irrelevant to the injected observables contribute little towards suppressing the dominant injection faults.

\subsection{Inject on Hypergraph Product Codes}

We also apply the construction to the $[[225,9,4]]$ HGP code of~\cite{xu_constant-overhead_2023} in order to test generality rather than to optimize performance. For the logical Pauli representatives, we use a CSS-specific construction that yields smaller-weight representatives than the generic stabilizer code construction of~\cite[Sec.~4.1]{gottesman_stabilizer_1997}, which is advantageous for injection. For syndrome extraction, we use the product-coloration circuit proposed in~\cite[Alg.~2]{xu_constant-overhead_2023}.

\begin{table}[htb]
    \centering
    \resizebox{\columnwidth}{!}{
\begin{tabular}{ccccccccc}
    \toprule
    $|\mathcal{I}|$ & $ p_{\cI}^{\rm tot} / |\mathcal{I}| $ & $ p_{\cI}^{\rm tot} $ & $ p^{\text{corr}}_\cI / |\mathcal{I}| $ & $ p^{\text{corr}}_\cI $ & $ p^{\text{corr}}_\cI / p_{\cI}^{\rm tot} $ & $|F|$ & $ r_{\rm disc}^{\rm UP}(0.0001) $ & $|\mathcal{P}|$ \\
    \midrule
    1 & 6.33$p$ & 6.33$p$ & 0 & 0 & 0 & 14 & 0.6278 & 0 \\
    2 & 6.50$p$ & 13.00$p$ & 0 & 0 & 0 & 19 & 0.6359 & 0 \\
    3 & 6.78$p$ & 20.33$p$ & 0 & 0 & 0 & 26 & 0.6413 & 0 \\
    4 & 7.25$p$ & 29.00$p$ & 0.08$p$ & 0.33$p$ & 0.01 & 25 & 0.6407 & 0 \\
    5 & 7.67$p$ & 38.33$p$ & 0.13$p$ & 0.67$p$ & 0.02 & 27 & 0.6443 & 0 \\
    6 & 8.00$p$ & 48.00$p$ & 0.17$p$ & 1.00$p$ & 0.02 & 28 & 0.6410 & 0 \\
    \bottomrule
\end{tabular}
}
\caption{
Analytical results for the best HGP code initial configuration found at each $|\mathcal{I}|$ under the uniform depolarizing noise model. Here $r_{\rm disc}^{\rm UP}$ is reported at $p=10^{-4}$ because the discard rate at $p=10^{-3}$ is above $0.9999$ for all listed configurations. The notation is otherwise the same as in \Cref{tab:ana-error-rate}.
}
\label{tab:ana-error-rate-hgp}
\end{table}

We could inject up to 6 logical qubits at once, and list the configuration with the smallest $p_{\cI}^{\rm tot}$ for each $|\mathcal{I}|$ in \Cref{tab:ana-error-rate-hgp}. Notably, none of these HGP configurations requires Bell states. The corresponding circuit-level simulation results under the uniform depolarizing model are shown in \Cref{fig:scalability-xu-hgp-li-uniform}.

\begin{figure}[htb]
    \centering
    \includegraphics[width=\columnwidth]{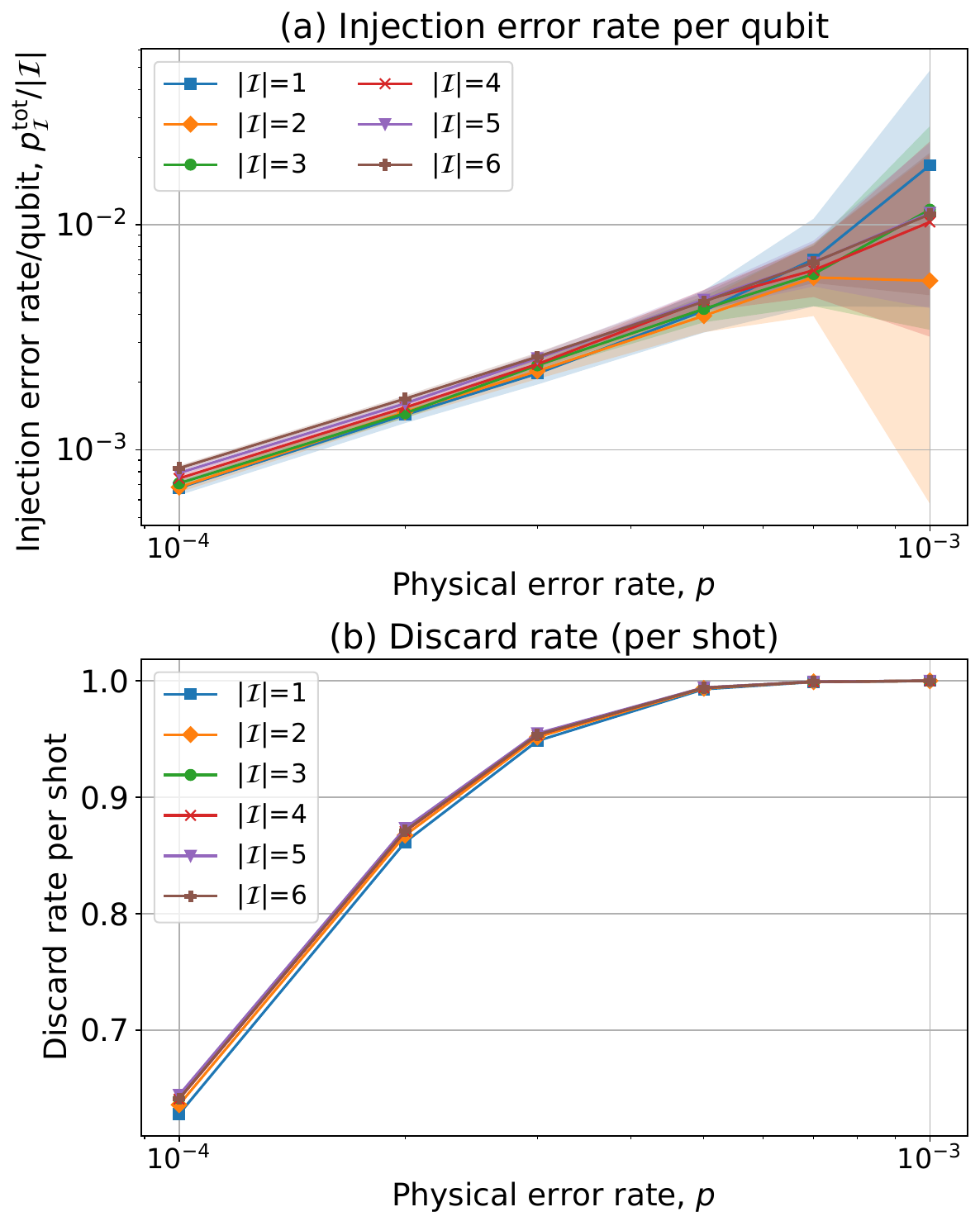}
    \caption{
        Circuit-level simulation on the HGP code $[[225,9,4]]$ from~\cite{xu_constant-overhead_2023} under uniform depolarizing noise model. (a) Injection error rate per logical qubit and (b) discard rate for the best configurations at each $|\cI|$ in \Cref{tab:ana-error-rate-hgp}. Each point uses $10^{7}$ shots.
    }
    \label{fig:scalability-xu-hgp-li-uniform}
\end{figure}

The current HGP results do not yet indicate near-term practicality: both the injection error rate and the discard rate are substantially worse than in the BB code case. We view this subsection primarily as a portability test. The present HGP study does not yet include code-specific optimization of logical representatives and syndrome extraction scheduling that tailor to injection. Since HGP codes cover a broad family of CSS qLDPC codes, these additional optimization directions make the HGP results a useful starting point rather than a final performance statement.

\subsection{Comparison with surface code injection}

To place our current BB code results in context, we compare against the hook injection scheme of Gidney~\cite{gidney_cleaner_2023}, which to our knowledge is the best-performing surface code magic state injection scheme at the raw injection level. This is a comparison of one-shot raw injection primitives, not of full cultivation or distillation pipelines. We compare four quantities: injection error rate, discard rate, discard-amortized spacetime volume, and physical qubit footprint.

Surface code hook injection is usually described as a two-step procedure: first inject a small surface code patch of distance $d_1$ using two rounds of SE, then enlarge it to a target distance $d_2$ using additional SE rounds.

\begin{table*}[htb]
    {\setlength{\tabcolsep}{8pt}
\begin{tabular}{lrrrrrr}
\toprule
$p$ & $10^{-4}$ & $2 \times 10^{-4}$ & $3 \times 10^{-4}$ & $5 \times 10^{-4}$ & $7 \times 10^{-4}$ & $10^{-3}$ \\
\midrule
Surf., $p_{\{1\}}$ & $5.3 \times 10^{-5}$ & $1.2 \times 10^{-4}$ & $1.8 \times 10^{-4}$ & $3.3 \times 10^{-4}$ & $4.8 \times 10^{-4}$ & $8.4 \times 10^{-4}$ \\
BB, $p_{\mathcal{I}}/|\mathcal{I}|$ & $1.5 \times 10^{-4}$ & $3.1 \times 10^{-4}$ & $4.6 \times 10^{-4}$ & $8 \times 10^{-4}$ & $1.1 \times 10^{-3}$ & $1.7 \times 10^{-3}$ \\
\midrule
Surf., $r_{\rm disc}$ & $0.04$ & $0.08$ & $0.11$ & $0.18$ & $0.25$ & $0.33$ \\
BB, $r_{\rm disc}$ & $0.24$ & $0.42$ & $0.55$ & $0.74$ & $0.85$ & $0.93$ \\
\midrule
Surf., $|\mathcal{I}|\cdot V^{\rm surf}$ & $5538$ & $5607$ & $5679$ & $5832$ & $5996$ & $6274$ \\
BB, $V^{\rm BB}$ & $7908$ & $10340$ & $13545$ & $23143$ & $39526$ & $88719$ \\
\midrule
Surf., \#Qubits & \multicolumn{6}{c}{$|\mathcal{I}|\cdot 2d_2^2 = 968$} \\
BB, \#Qubits & \multicolumn{6}{c}{$288$} \\
\bottomrule
\end{tabular}
}
\caption{
Comparison between Gidney's surface code hook injection and our BB code injection for the best $|\mathcal{I}|=4$ BB configuration under the uniform depolarizing noise model. We report the single logical surface code injection error rate $p_{\{1\}}$, the approximate BB injection error rate per logical qubit $p_{\mathcal{I}}/|\mathcal{I}|$, the discard rate $r_{\rm disc}$, the discard-amortized spacetime volume $V$, and the physical qubit count. For the surface code benchmark, we inject a $d_1=5$ patch using two SE rounds with post-selection on fixed stabilizers and then enlarge to $d_2=11$ with one additional SE round; decoding uses \texttt{pymatching}. For the BB code benchmark, we inject the full $[[144,12,12]]$ block using the same total number of SE rounds; decoding uses BP+OSD.
}
\label{tab:cmp-surface-bb}
\end{table*}

To choose a comparable target distance, we follow the memory-level comparison in~\cite[Fig.~2]{bravyi_high-threshold_2024}, which indicates that one $[[144,12,12]]$ BB block has memory performance similar to 12 surface code patches of distance 11 or 13. We therefore use $d_2=11$ as the target surface code distance in this benchmark. On the BB side, we benchmark the best $|\cI|=4$ configuration because its correlated injection error rate remains negligible.

Let $r_{\rm disc}$ be the discard rate under post-selection on the fixed stabilizers, $V_{\rm inj}$ the spacetime volume of the injection procedure, and $V_{\rm grow}$ the spacetime volume of the code-enlargement procedure. We define spacetime volume as ``number of physical qubits'' times ``number of two-qubit gates on the critical path.'' The discard-amortized spacetime volume is then
\begin{equation}
V = \frac{V_{\rm inj}}{1 - r_{\rm disc}} + V_{\rm grow}.
\label{eq:spacetime-volume}
\end{equation}

The comparison is summarized in \Cref{tab:cmp-surface-bb}. At present, the $[[144,12,12]]$ BB code injection scheme is not yet competitive with surface code hook injection in either injection error rate per logical qubit or discard-amortized spacetime volume. The main reason is that the entire BB block must be injected at once, which drives the discard rate high enough to outweigh the benefit of injecting multiple logical qubits simultaneously. 

\section{Conclusion and future directions}
\label{sec:conclusion}

We presented a general in-situ and simultaneous magic state injection method for CSS qLDPC codes. Our scheme prepares multiple logical magic states directly within a single qLDPC block using only standard syndrome extraction resources, without auxiliary preparation blocks for state transfer. This work establishes a qLDPC-native injection primitive for encoded non-Clifford resource states. We showed that the same framework can be substantially improved by optimizing the initial configuration to maximize first-round protection of the injected logical supports. The main advantage of our scheme is the qubit footprint: the BB injection uses 288 physical qubits, compared with 968 for four distance-11 surface code patches. This substantially smaller footprint makes the BB approach relevant for footprint-constrained demonstrations, especially for early fault-tolerant architectures. Circuit-level evaluation on BB and HGP code instances demonstrates both practical operating points on selected BB configurations and general applicability beyond a single code family. 

Several directions remain open. First, increasing the number of simultaneously injectable logical qubits while retaining low error rates will likely require stronger combinatorial conditions, better logical representatives, and improved syndrome extraction circuits. Second, reducing the discard rate remains important for practical throughput; our results suggest that this will require code- and circuit-level improvements, especially in the structure of the round-parity detectors. Finally, a hardware-aware version of the scheme should adapt Bell state preparation on overlapping qubits to the native connectivity, gate set, and resource constraints of a given qLDPC memory experiment. We hope these directions will lead to fully qLDPC-native protocols for high-quality non-Clifford resource state preparation.

\begin{acknowledgments}
We thank Hengyun (Harry) Zhou and J. Pablo Bonilla Ataides for providing the data of hypergraph product codes.
We thank Patrick Rall for valuable insights and discussions.
This work is supported in part by the National Science Foundation (under awards CCF-2312754 and CCF-2338063), in part by the Department of Energy Co-Design Center for Quantum Advantage (C2QA), in part by QuantumCT (under NSF Engines award ITE-2302908), in part by AFOSR MURI (FA9550-26-1-B036).
Z.H. acknowledges support from the MIT Department of Mathematics, the MIT-IBM Watson AI Lab, and the NSF Graduate Research Fellowship Program under Grant No. 2141064.
YD acknowledges partial support by Boehringer Ingelheim, and NSF NQVL-ERASE (under award OSI-2435244). External interest disclosure: YD is a scientific advisor to D-Wave Quantum, Inc.
\end{acknowledgments}

\section*{DATA AVAILABILITY}
The data and code used to generate 
the results in this paper is available on GitHub~\footnote{https://github.com/kunliu7/qldpc-msi}.
\appendix

\section{High-rate qLDPC codes}
\label{app:high-rate-qLDPC-codes}

In this section, we give the definitions of the qLDPC codes used in this work.

\paragraph{Bivariate Bicycle (BB) codes}
The construction of BB codes~\cite{bravyi_high-threshold_2024,yoder_tour_2025} begins with the $\ell\times\ell$ cyclic shift matrix $S_\ell$, where the $i$-th row contains a single nonzero entry in column $(i+1) \bmod \ell$.
For example, $S_3$ is given by:
\begin{equation}
    S_3 = \begin{pmatrix}
        0 & 1 & 0 \\
        0 & 0 & 1 \\
        1 & 0 & 0
    \end{pmatrix}.
\end{equation}
We define two generator matrices $x = S_\ell \otimes I_m$ and $y = I_\ell \otimes S_m$, where $\otimes$ denotes the Kronecker product.
These generators satisfy $x^{\ell} = y^{m} = I_{\ell m}$, and any monomial $x^p y^q$ (for $0 \leq p < \ell$ and $0 \leq q < m$) corresponds to the matrix $S_\ell^p \otimes S_m^q$, allowing us to treat monomials and matrices interchangeably.
We define the transpose of $x^p y^q$ to be $(x^p y^q)^\top = x^{\ell -p} y^{m - q} = x^{-p} y^{-q}$.

A BB code is parameterized by two polynomials $A$ and $B$, each expressed as a sum of three distinct monomials:
\begin{equation}
    A = A_1 + A_2 + A_3, \quad B = B_1 + B_2 + B_3,
\end{equation}
where each $A_i$ and $B_i$ is a distinct monomial of the form $x^p y^q$.
The parity-check matrices for the $X$ and $Z$ stabilizers are then constructed as:
\begin{equation}
H_X = [A|B], \quad H_Z = [B^\top|A^\top],
\end{equation}
where $H_X$ and $H_Z$ are $n/2 \times n$ matrices with $n = 2 \ell m$ qubits.

A widely known example of BB codes is the $[[144,12,12]]$ code,
which is defined by $ \ell = 12 $, $ m = 6 $, and the polynomials:
\begin{equation}
    A=1+y+x^{3} y^{-1}, \quad B=1+x+x^{-1} y^{-3}.
\end{equation}
We use the logical Pauli representatives~\cite[Sec.~A.1]{yoder_tour_2025}
that has undergone a series of optimizations~\cite{bravyi_high-threshold_2024,
cross_improved_2024,williamson_low-overhead_2024,yoder_tour_2025}
targeting logical operations with low-overhead ancilla system.
Fortunately, such representatives are also friendly to injection,
where overlaps are intentionally designed to be limited,
and the weight of the logical Pauli representatives is exactly $ d=12 $.

\paragraph{Hypergraph Product (HGP) Codes}

We also inject on the $[[225,9,4]]$ HGP code 
from~\cite{xu_constant-overhead_2023} to demonstrate the generality of our scheme.
Such family of HGP codes are constructed as the product of classical LDPC codes defined by $(3,4)$-regular Tanner graphs, and have good circuit-level performance with constant overhead~\cite{xu_constant-overhead_2023},
and fast, parallelizable logical computation~\cite{xu_fast_2024}.

Denote the parity check matrices of the two codes as $H_1$ and $H_2$, with shape $ r_i \times n_i $ respectively.
The parity check matrix of the HGP code is given by:

\begin{equation}
    \begin{aligned}
    & H_X=\left(\begin{array}{ll}
    H_1 \otimes I_{n_2} & I_{r_1} \otimes H_2^\top
    \end{array}\right) \\
    & H_Z=\left(\begin{array}{ll}
    I_{n_1} \otimes H_2 & H_1^\top \otimes I_{r_2}
    \end{array}\right),
    \end{aligned}
\end{equation}
where $ I_{m} $ is the $ m $-dimensional identity matrix.

The standard logical Pauli representatives construction for QEC codes~\cite[Sec.~4.1]{gottesman_stabilizer_1997}
don't guarantee short-weight logical Pauli representatives for all the logical qubits.
For HGP codes, we construct the logical Pauli representatives with shorter weight than the generic construction.
We use the syndrome extraction circuit in Ref.~\cite{xu_constant-overhead_2023}.

\section{Maximal injectable sets}
\label{sec:maximal-injectable-sets}

An index set $\cI \subseteq [k]$ is injectable under injection sites
$\bm Q=\{Q_1,\dots,Q_k\}$ with $Q_i\in O_i$ if it satisfies
injection site independence (\Cref{def:injection-site-independence})
and compatible pairing (\Cref{def:generic-compatible}).

We compute injectable sets by enumerating injection-site choices and then checking subset feasibility:
\begin{enumerate}
\item Enumerate the Cartesian product $\prod_i O_i$ to obtain candidate $\bm Q$.

\item For fixed $\bm Q$, build a graph on vertices $\{1,\dots,k\}$, where edge $(i,j)$ means $i,j$ cannot be injected together because injection site independence is violated.
Any injectable set must be an independent set in this graph.
Compute all maximal independent sets (MISs) as candidates.

\item For each MIS candidate, enumerate subsets and test compatible pairing using signatures.
For each physical qubit $q$ in any $C_{i,j}$ (defined in \Cref{eq:unified-Cij}), define its signature
\begin{equation}
\sigma(q) \;:=\; \bigl(\mathbf 1[q\in C_{i,j}]\bigr)_{i,j\in\cI}.
\label{eq:signature-def}
\end{equation}
Compatible pairing exists iff every signature class $\{q:\sigma(q)=s\}$ has even cardinality.
Within each signature class, any disjoint pairing is valid.
\end{enumerate}

\section{Noise model}

The noise model we used in this work is the same as the noise model in~\cite{li_magic_2015,lao_magic_2022}, as shown in \Cref{tab:li_noise_channels,tab:li_noise_model}.

\begin{table}[t]
    \centering
    \begin{tabular}{c|l}
    \toprule
    \textbf{Noise channel} & \textbf{Definition} \\
    \midrule
    $\mathrm{MERR}(p)$ &
    $m \mapsto m\oplus e,\ \ e\sim \mathrm{Bernoulli}(p)$ \\
    $\mathrm{XERR}(p)$ &
    $\rho \mapsto (1-p)\rho + p\,X\rho X$ \\
    $\mathrm{ZERR}(p)$ &
    $\rho \mapsto (1-p)\rho + p\,Z\rho Z$ \\
    $\mathrm{DEP1}(p)$ &
    $\rho \mapsto (1-p)\rho + \frac{p}{3}\sum_{P\in\{X,Y,Z\}} P\rho P$ \\
    $\mathrm{DEP2}(p)$ &
    $\rho \mapsto (1-p)\rho + \frac{p}{15}\!\!\sum_{P\in\{I,X,Y,Z\}^{\otimes2}\setminus\{II\}}\!\! P\rho P$ \\
    \bottomrule
\end{tabular}
\caption{Definitions of the noise channels.}
\label{tab:li_noise_channels}
\end{table}

\begin{table}[t]
\centering
\renewcommand{\arraystretch}{1.2}
\begin{tabular}{l|l}
\toprule
\textbf{Ideal operation} & \textbf{Noisy operation} \\
\midrule
$R_X$ (prepare ${+}1$ eigenstate of $X$)
& $R_X \circ \mathrm{ZERR}(p_{\rm IN})$ \\
$R_Y$ (prepare ${+}1$ eigenstate of $Y$)
& $R_Y \circ \mathrm{XERR}(p_{\rm IN})$ \\
$R_Z$ (prepare ${+}1$ eigenstate of $Z$)
& $R_Z \circ \mathrm{XERR}(p_{\rm IN})$ \\
$M_Z$ (measure $Z$)
& $M_Z \circ \mathrm{MERR}(p_{\rm M})$ \\
$M_X$ (measure $X$)
& $M_X \circ \mathrm{MERR}(p_{\rm M})$ \\
$U_1$ (any single-qubit unitary)
& $U_1 \circ \mathrm{DEP1}(p_{\rm 1Q})$ \\
$U_2$ (any two-qubit unitary)
& $U_2 \circ \mathrm{DEP2}(p_{\rm 2Q})$ \\
$\text{IDLE}$ (idle)
& $\text{IDLE} \circ \mathrm{DEP1}(p_{\rm IDLE})$ \\
\bottomrule
\end{tabular}
\caption{The noise model we used in this work, the same as circuit-level noise model in~\cite{li_magic_2015,lao_magic_2022}. The notation $G \circ \mathcal{E}$ indicates that the ideal operation $G$ is \emph{followed} by the noise channel $\mathcal{E}$.}
\label{tab:li_noise_model}
\end{table}

\section{General k-injection details}
\label{app:k-injection-details}

In this section, we give the stabilizer-cleaning construction used in \Cref{sec:k-injection} and explain why it preserves the encoded logical information.

\subsection{Logical representatives reduction}
We describe a recursive algorithm for transforming the logical representatives to a set of $ k $ physical qubits.
Let
\begin{equation}
\mathcal{L} := \{ \bar{X}_i, \bar{Z}_i \}_{i=1}^{k}
\label{eq:tomas-current-logical-basis}
\end{equation}
be the current logical representatives, and let $\mathcal{S}$ be the current stabilizer generator
set.

The recursive reduction is as follows. Given a non-trivial stabilizer group, 
do the following:
\begin{enumerate}
    \item Pick a stabilizer $s \in \mathcal{S}$ with $\operatorname{wt}(s) > 1$.
    \item Measure a single-qubit Pauli $P_j$ that anticommutes with $s$ on a qubit $j \in \operatorname{supp}(s)$, i.e.,
    \begin{equation}
    \{ P_j, s \} = 0;
    \label{eq:tomas-step-anticommute}
    \end{equation} denote the measurement outcome by $m \in \{ \pm 1 \}$.
    \item Replace $s \in \mathcal{S}$ by the measured $P_j$.
    \item Update any remaining stabilizers and logical representatives that
    anticommute with $P_j$ by multiplying them by the replaced stabilizer $s$.
\end{enumerate}

To write this explicitly, let
\begin{equation}
\mathcal{S} \setminus \{ s \} = \{ g_a \}_{a=1}^{n-k-1} .
\label{eq:tomas-step-other-stabilizers}
\end{equation}
The updated generator set is
\begin{equation}
\mathcal{S}' := \{ m P_j \} \cup \{ g_a' \}_{a=1}^{n-k-1},
\label{eq:tomas-step-stabilizer-set}
\end{equation}
where
\begin{equation}
g_a' =
\begin{cases}
g_a, & [g_a, P_j] = 0, \\[3pt]
g_a s, & \{ g_a, P_j \} = 0 .
\end{cases}
\label{eq:tomas-step-stabilizer-update}
\end{equation}
Similarly, each logical representative $L \in \mathcal{L}$ is updated to
$L' \in \mathcal{L}'$ according to
\begin{equation}
L' =
\begin{cases}
L, & [L, P_j] = 0, \\[3pt]
L s, & \{ L, P_j \} = 0 .
\end{cases}
\label{eq:tomas-step-logical-update}
\end{equation}
We denote the updated logical representatives by
\begin{equation}
\mathcal{L}' := \{ \bar{X}_i', \bar{Z}_i' \}_{i=1}^{k}.
\label{eq:tomas-step-updated-logical-basis}
\end{equation}

Each reduction ``peels off'' a physical qubit $j$ from the original logical representatives $\mathcal{L}$, since qubit $j$ is fixed by a weight-$1$ Pauli $P_j$.
We repeat the above reduction step
until all stabilizers are weight-$1$, and obtain
\begin{equation}
\mathcal{S} \to \mathcal{S}', \quad
\mathcal{L} \to \mathcal{L}',
\end{equation}
where $ \mathcal{S}' $ contains only weight-$1$ Pauli operators,
and $ \mathcal{L}' $ is the new set of logical representatives that only live on a set of $ k $ physical qubits. 
We call these $ k $ physical qubits the \emph{carrier qubits} and denote them by $\mathcal{Q}_L = \set{Q_i}_{i=1}^k$, and call the remaining $ n-k $ physical qubits supporting single-qubit Paulis in $ \mathcal{S}' $ the \emph{peeled qubits}, $\mathcal{Q}_P$.

Since $ \mathcal{L}' $ is a valid conjugate logical representative,
there exists a Clifford $C$ such that it maps the standard computational basis on the $ k $ carrier qubits to the new logical representatives $\mathcal{L}'$:
\begin{equation}
C X_{Q_i} C^\dagger = \bar{X}_i',
\qquad
C Z_{Q_i} C^\dagger = \bar{Z}_i'.
\end{equation}
And thus, $C Y_{Q_i} C^\dagger = \bar{Y}_i'$.

Therefore, to prepare the $+1$ eigenstate of $ \bar{M}_i' $ for all $ i \in [k] $, we first prepare the carrier qubits $\set{Q_i}$ in the $+1$ eigenstates of $ \{M_{Q_i}\} $, then apply the Clifford circuit $C$ to the carrier qubits:
\begin{equation}
\begin{aligned}
    \bar{M}_i' &= \cos\theta_i\, \bar{X}_i' + \sin\theta_i\, \bar{Y}_i' \\
    &= C \bigl(\cos\theta_i\, X_{Q_i} + \sin\theta_i\, Y_{Q_i}\bigr) C^\dagger \\
    &= C \;{M}_{Q_i}\; C^\dagger .
\end{aligned}
\end{equation}

Although the construction above applies generally to arbitrary CSS codes and can inject all \(k\) logical qubits simultaneously, it has several caveats.
One limitation is that the state-preparation Clifford circuit \(C\) is typically nonlocal and deep, and therefore not itself fault tolerant. 
This issue may be mitigated by implementing or compiling \(C\) on the logical level instead of the physical level. 
For instance, on an extractor architecture~\cite{he_extractors_2025} implementing Pauli-based computation~\cite{Bravyi2016}, a logical Clifford circuit can be implemented by changing the Pauli basis of subsequent logical Pauli-rotations. 
In other words, instead of implementing \(C\), we simply change subsequent gates from \(e^{i\theta P}\) to \(e^{i\theta C^\dagger PC}\). 
In a model where every Pauli rotation has roughly the same implementation cost~\cite{he_extractors_2025}, this Clifford rotation is essentially free.

\subsection{Logical information preservation}

During the reduction step, the update rule in
\Cref{eq:tomas-step-stabilizer-update,eq:tomas-step-logical-update}
preserves the logical subsystem for three reasons.

First, the measurement of $P_j$ reveals no logical information. Let
$\ket{\phi}$ be any state in the codespace of $S$, so in particular
$s \ket{\phi} = \ket{\phi}$. Using \Cref{eq:tomas-step-anticommute},
\begin{equation}
\bra{\phi} P_j \ket{\phi}
=
\bra{\phi} s P_j s \ket{\phi}
=
- \bra{\phi} P_j \ket{\phi}
= 0 .
\label{eq:tomas-step-random-outcome}
\end{equation}
Hence the two outcomes $m = \pm 1$ occur with equal probability $1/2$,
independently of the encoded logical state. The measurement therefore creates a
new local stabilizer sign, but it does not read out any logical observable.

Second, the number of encoded qubits is unchanged. The original generator $s \in \mathcal{S}$ is
removed and replaced by the new independent generator $m P_j$. Thus
$\mathcal{S}'$ again contains $n-k$ independent commuting generators, so the
updated code still encodes exactly $k$ logical qubits.

Third, the logical representatives are changed only by multiplication with the original
stabilizer $s$. If $L' = L s$, then on the original codespace,
\begin{equation}
L' \ket{\phi} = L s \ket{\phi} = L \ket{\phi} .
\label{eq:tomas-step-same-logical-action}
\end{equation}
Thus \Cref{eq:tomas-step-logical-update} changes only the
\emph{representatives} of the logical Paulis, not the logical Pauli classes
themselves.
And thus, 
\begin{equation}
[\bar{M}_i', \proj_{\mathcal{S}}] = 0 \quad \forall i \in [k].
\end{equation}
In another word, if the state is prepared in the simultaneous $ +1 $ eigenstate of $ \bar{M}_i' $ for all $ i \in [k] $,
then after the SE of $\mathcal{S}$, the state is the simultaneous $ +1 $ eigenstate of $ \bar{M}_i $ for all $ i \in [k] $.

\section{Proof of \Cref{thm:unified-sufficient-condition}}
\label{sec:proof-all-in-one}

\begin{proof}[Proof sketch]
    For each injected logical $i$, decompose the logical representatives into two parts:
    \[
    \bar X_i = X_{O_i}X_{A_i},
    \qquad
    \bar Z_i = Z_{O_i}Z_{B_i}.
    \]
    And thus, assume $ |O_i| = 2t+1 $, we have
    \[
    \bar Y_i = (-1)^t Y_{O_i}X_{A_i}Z_{B_i}.
    \]
    The proof has two corresponding conceptual parts: first on $ A_i $ and $ B_i $ together, and second on $ O_i $ only.
    
    First, the prescribed preparation enforces a family of stabilizer constraints.
    The Bell states on the compatible pairing $\cP$ imply that every overlap set
    $C_{i,j}$ is stabilized by both $X_{C_{i,j}}$ and $Z_{C_{i,j}}$.
    Together with the direct preparation of $\ket{+}$ on $A_i\setminus C$,
    $\ket{0}$ on $B_i\setminus C$, and the $+1$ eigenstate of $M_{Q_i}$ on $Q_i$,
    this yields
    \begin{equation}
    X_{A_i}|\psi_{\mathrm{in}}\rangle = |\psi_{\mathrm{in}}\rangle, \quad
    Z_{B_i}|\psi_{\mathrm{in}}\rangle = |\psi_{\mathrm{in}}\rangle, \quad
    \forall i\in\cI.
    \label{eq:multi-prep-generic}
    \end{equation}

    Second, these constraints let us peel the logical support away from the
    injection site.
    The factors on $A_i$ and $B_i$ act trivially, so $\bar X_i$ and $\bar Y_i$
    reduce to operators on $O_i$ only.
    Then the Bell pairs inside $O_i\setminus\{Q_i\}$ cancel pairwise, leaving only
    the single-weight operator on $Q_i$:
    \begin{equation}
    \begin{aligned}
    \bar M_i|\psi_{\mathrm{in}}\rangle 
    &= \left(\cos\theta_i\,X_{O_i} X_{A_i} + \sin\theta_i\,Y_{O_i} X_{A_i} Z_{B_i}\right) |\psi_{\mathrm{in}}\rangle \\
    &= \left(\cos\theta_i\,X_{O_i} + \sin\theta_i\,Y_{O_i}\right) |\psi_{\mathrm{in}}\rangle \\
    &= M_{Q_i}|\psi_{\mathrm{in}}\rangle = |\psi_{\mathrm{in}}\rangle,
    \qquad \forall i\in\cI.
    \end{aligned}
    \label{eq:M-reduce-to-Qi-app}
    \end{equation}
\end{proof}

Following the proof sketch, we prove in two parts.
In Part 1 of the proof, we show \Cref{eq:multi-prep-generic}.
In Part 2 of the proof, we show \Cref{eq:M-reduce-to-Qi-app}.

\begin{proof}[Proof, Part 1]
\textbf{Step 1: Record the stabilizers imposed directly by the preparation.}
For each pair $\{u,v\}\in\cP$, let
\[
|\mathrm{Bell}_{\{u,v\}}\rangle := (|00\rangle+|11\rangle)/\sqrt2.
\]
By construction, the prepared state $|\psi_{\mathrm{in}}\rangle$ satisfies
\begin{equation}
    \begin{aligned}
(X_uX_v)|\psi_{\mathrm{in}}\rangle &= |\psi_{\mathrm{in}}\rangle, \\
(Z_uZ_v)|\psi_{\mathrm{in}}\rangle &= |\psi_{\mathrm{in}}\rangle, \quad
\forall \{u,v\}\in\cP.
\end{aligned}
\label{eq:bell-pair-stabs-app}
\end{equation}
Also, for every $i\in\cI$,
\begin{equation}
\begin{aligned}
X_{A_i\setminus C}|\psi_{\mathrm{in}}\rangle &= |\psi_{\mathrm{in}}\rangle, \\
Z_{B_i\setminus C}|\psi_{\mathrm{in}}\rangle &= |\psi_{\mathrm{in}}\rangle, \\
M_{Q_i}|\psi_{\mathrm{in}}\rangle &= |\psi_{\mathrm{in}}\rangle.
\end{aligned}
\label{eq:direct-prep-stabs-app}
\end{equation}
Injection site independence ensures that $Q_i$ is not required by any other
logical support, so preparing the $+1$ eigenspace of $M_{Q_i}$ is compatible
with the rest of the construction.

\medskip
\noindent
\textbf{Step 2: Use compatible pairing to control every overlap set $C_{i,j}$.}
Fix any $i,j\in\cI$.
Because $\cP$ is compatible, each Bell pair is either fully contained in
$C_{i,j}$ or disjoint from it.
Hence the restricted set
\[
\cP_{i,j}:=\{\{u,v\}\in\cP:\{u,v\}\subseteq C_{i,j}\}
\]
partitions $C_{i,j}$.
Therefore,
\[
X_{C_{i,j}}=\prod_{\{u,v\}\in\cP_{i,j}} X_uX_v,
\qquad
Z_{C_{i,j}}=\prod_{\{u,v\}\in\cP_{i,j}} Z_uZ_v.
\]
Using \Cref{eq:bell-pair-stabs-app}, we obtain
\begin{equation}
X_{C_{i,j}}|\psi_{\mathrm{in}}\rangle = |\psi_{\mathrm{in}}\rangle,
\qquad
Z_{C_{i,j}}|\psi_{\mathrm{in}}\rangle = |\psi_{\mathrm{in}}\rangle,
\qquad \forall i,j\in\cI.
\label{eq:cij-stabs-app}
\end{equation}

\medskip
\noindent
\textbf{Step 3: Extend from overlaps to the full sets $A_i$ and $B_i$.}
We next show that \Cref{eq:multi-prep-generic} hold.

Fix $i\in\cI$.
Decompose
\[
A_i=(A_i\setminus C)\sqcup(A_i\cap C),
\qquad
B_i=(B_i\setminus C)\sqcup(B_i\cap C).
\]
The factors on $A_i\setminus C$ and $B_i\setminus C$ are already handled by
\Cref{eq:direct-prep-stabs-app}, so it remains to treat the parts inside $C$.

Consider $A_i\cap C$.
Any qubit $u\in A_i\cap C$ lies in some overlap set
$C_{i,j}=A_i\cap B_j$.
Since $\cP$ is compatible, the partner of $u$ under $\cP$ lies in the same
$C_{i,j}$, hence also in $A_i\cap C$.
Therefore the Bell pairs that intersect $A_i\cap C$ actually partition
$A_i\cap C$, so
\[
X_{A_i\cap C}
=
\prod_{\{u,v\}\subseteq A_i\cap C} X_uX_v.
\]
Every factor on the right stabilizes $|\psi_{\mathrm{in}}\rangle$ by
\Cref{eq:bell-pair-stabs-app}, hence
$X_{A_i\cap C}|\psi_{\mathrm{in}}\rangle=|\psi_{\mathrm{in}}\rangle$.
Combining with \Cref{eq:direct-prep-stabs-app} yields
\begin{equation}
X_{A_i}|\psi_{\mathrm{in}}\rangle = |\psi_{\mathrm{in}}\rangle.
\label{eq:Ai-stab-app}
\end{equation}

The same argument applies to $B_i\cap C$.
Indeed, any $u\in B_i\cap C$ lies in some $C_{k,i}=A_k\cap B_i$, and
compatibility forces its partner to lie in the same set, hence also in
$B_i\cap C$.
Thus
\[
Z_{B_i\cap C}
=
\prod_{\{u,v\}\subseteq B_i\cap C} Z_uZ_v,
\]
so $Z_{B_i\cap C}$ stabilizes $|\psi_{\mathrm{in}}\rangle$, and together with
\Cref{eq:direct-prep-stabs-app} we get
\begin{equation}
Z_{B_i}|\psi_{\mathrm{in}}\rangle = |\psi_{\mathrm{in}}\rangle.
\label{eq:Bi-stab-app}
\end{equation}

At this point, \Cref{eq:multi-prep-generic} is established.
\medskip
\noindent
\end{proof}

\begin{proof}[Proof, Part 2]
\textbf{Step 4: Reduce $\bar X_i$ and $\bar Y_i$ to operators on $O_i$.}
Fix $i\in\cI$.
Recall the CSS representatives
\begin{equation}
\bar X_i = X_{O_i}X_{A_i},
\qquad
\bar Z_i = Z_{O_i}Z_{B_i}.
\label{eq:css-reps-app}
\end{equation}
Using \Cref{eq:Ai-stab-app}, we can drop the $A_i$ factor when acting on
$|\psi_{\mathrm{in}}\rangle$:
\begin{equation}
\bar X_i|\psi_{\mathrm{in}}\rangle = X_{O_i}|\psi_{\mathrm{in}}\rangle.
\label{eq:X-reduce-to-Oi-app}
\end{equation}

Next, write $|O_i|=2t+1$.
Then
\[
\bar Y_i := i\bar X_i\bar Z_i = (-1)^t Y_{O_i}X_{A_i}Z_{B_i}.
\]
Applying \Cref{eq:Ai-stab-app,eq:Bi-stab-app}, we obtain
\begin{equation}
\bar Y_i|\psi_{\mathrm{in}}\rangle = (-1)^t Y_{O_i}|\psi_{\mathrm{in}}\rangle.
\label{eq:Y-reduce-to-Oi-app}
\end{equation}

So the logical rotated observable has already been reduced from the full
logical support to the core set $O_i$.

\medskip
\noindent
\textbf{Step 5: Peel off the paired qubits in $O_i\setminus\{Q_i\}$.}
We now use the Bell-pair structure inside $O_i\setminus\{Q_i\}$ to reduce from
$O_i$ all the way down to the single site $Q_i$.

By compatibility, the restricted set
\[
\cP_{i,i}:=\{\{u,v\}\in\cP:\{u,v\}\subseteq O_i\setminus\{Q_i\}\}
\]
partitions $O_i\setminus\{Q_i\}$.
Since $|O_i|=2t+1$, we may write
\[
O_i\setminus\{Q_i\}=\bigsqcup_{\ell=1}^t \{r_\ell,s_\ell\}.
\]
For each pair $\{r_\ell,s_\ell\}$, \Cref{eq:bell-pair-stabs-app} implies
\[
(X_{r_\ell}X_{s_\ell})|\psi_{\mathrm{in}}\rangle=|\psi_{\mathrm{in}}\rangle,
\qquad
(Z_{r_\ell}Z_{s_\ell})|\psi_{\mathrm{in}}\rangle=|\psi_{\mathrm{in}}\rangle.
\]
Multiplying these two relations gives
\[
(-Y_{r_\ell}Y_{s_\ell})|\psi_{\mathrm{in}}\rangle=|\psi_{\mathrm{in}}\rangle.
\]
Therefore the paired qubits cancel in both the $X$ and $Y$ channels:
\begin{equation}
X_{O_i}|\psi_{\mathrm{in}}\rangle = X_{Q_i}|\psi_{\mathrm{in}}\rangle,
\qquad
Y_{O_i}|\psi_{\mathrm{in}}\rangle = (-1)^t Y_{Q_i}|\psi_{\mathrm{in}}\rangle.
\label{eq:O-to-Q-app}
\end{equation}

Combining \Cref{eq:X-reduce-to-Oi-app,eq:Y-reduce-to-Oi-app,eq:O-to-Q-app},
we conclude
\[
\bar X_i|\psi_{\mathrm{in}}\rangle = X_{Q_i}|\psi_{\mathrm{in}}\rangle,
\qquad
\bar Y_i|\psi_{\mathrm{in}}\rangle = Y_{Q_i}|\psi_{\mathrm{in}}\rangle.
\]
Hence
\[
\bar M_i(\theta_i)|\psi_{\mathrm{in}}\rangle
=
\bigl(\cos\theta_i\,X_{Q_i}+\sin\theta_i\,Y_{Q_i}\bigr)|\psi_{\mathrm{in}}\rangle
=
M_{Q_i}|\psi_{\mathrm{in}}\rangle.
\]
Finally, \Cref{eq:direct-prep-stabs-app} gives
$M_{Q_i}|\psi_{\mathrm{in}}\rangle=|\psi_{\mathrm{in}}\rangle$.
Therefore
\[
\bar M_i|\psi_{\mathrm{in}}\rangle = |\psi_{\mathrm{in}}\rangle,
\qquad \forall i\in\cI.
\]

In summary, the preparation supplies local and pairwise stabilizers,
compatible pairing ensures that all relevant support away from $Q_i$ is
organized into cancellable Bell pairs, and thus each logical rotated
observable reduces exactly to the prepared single-site operator at $Q_i$.
\end{proof}

\section{Detector error model}
\label{sec:detector-error-model}

We use the detector error model (DEM)~\cite{gidney_stim_2021} as a compact description of the noisy syndrome extraction circuit.
For our purposes, the DEM serves two roles.
First, it provides a convenient first-order language for identifying harmful error mechanisms that survive postselection and induce logical injection errors.
Second, it yields a simple explanatory proxy for the discard rate, which is one of the main quantities entering the spacetime-overhead analysis.
This appendix also explains why the discard rate remains similar across different optimized initial configurations.

\subsection{Definitions}

A \emph{detector} is a binary parity check built from a specified subset of measurement-record bits.
Detectors are defined so that, in the absence of noise, every detector outcome is deterministically $0$.
Noise can toggle some of these detector outcomes to $1$.

A DEM records how elementary stochastic mechanisms toggle detector outcomes and tracked logical observables, without explicitly tracking the full quantum state.
For the analysis below, we represent the DEM as a collection of Bernoulli mechanisms indexed by $e$.
Each mechanism is specified by a triple
\begin{equation}
(p_e, D_e, L_e),
\label{eq:dem_mechanism_triple}
\end{equation}
where $p_e\in[0,1]$ is the occurrence probability of mechanism $e$, $D_e\subseteq\{1,\dots,N_D\}$ is the set of detectors toggled by $e$, and $L_e\subseteq\{1,\dots,N_L\}$ is the set of tracked logical observables toggled by $e$.
Here $N_L$ is the number of logical observables included in the DEM analysis, and $N_D$ is the number of detectors.

\subsection{Error mechanism counting}

For MSI, the tracked logical observables are the injected logical observables indexed by $\cI$, so we take $L_e\subseteq \cI$.
The most dangerous first-order mechanisms are those with
\begin{equation}
L_e \neq \emptyset,
\qquad
D_e = \emptyset,
\end{equation}
because they flip one or more injected logical observables without triggering any detector.
Such mechanisms evade postselection and also leave no detector information for decoding.

Because detector and logical outcomes combine by parity, multiple mechanisms can cancel each other.
However, the probability that two distinct mechanisms $e_1$ and $e_2$ occur in the same shot is second order, namely $p_{e_1}p_{e_2}$.
Accordingly, a standard first-order approximation is to assume that at most one DEM mechanism occurs during a shot.
Previous MSI analyses use this viewpoint to count undetected logical-fault mechanisms and find good agreement with circuit-level Monte Carlo simulation~\cite{li_magic_2015,lao_magic_2022,gidney_cleaner_2023}.
In our setting, this yields
\begin{equation}
p_\cJ \approx \sum_{e:\,L_e = \cJ} p_e,
\end{equation}
where $p_\cJ$ is the probability that the incorrect logicals are \emph{exactly} the subset $\cJ$.
In \Cref{sec:evaluation}, we show that this first-order approximation closely matches circuit-level Monte Carlo simulation.

\subsection{Postselection and discard rate}

Let $P\subseteq\{1,\dots,N_D\}$ be the set of postselection detectors.
In our simulations, $P$ consists of (i) fixed-stabilizer detectors from the first SE round and (ii) round-parity detectors comparing the first two SE rounds.
A shot is \emph{accepted} if and only if all detectors in $P$ are $0$; otherwise the shot is discarded.

Define the set of \emph{postselection-relevant} mechanisms by
\begin{equation}
R := \{e:\; D_e\cap P\neq \emptyset\}.
\label{eq:relevant_mechanisms_set_simplified}
\end{equation}
If no mechanism in $R$ occurs, then every postselection detector remains $0$.
Therefore, discard implies that at least one mechanism in $R$ occurred, which gives the upper bound
\begin{equation}
r_{\rm disc}
\le
\Pr\!\left(\exists e\in R:\; e\text{ occurs}\right).
\end{equation}
The inequality above is generally strict because multiple postselection-relevant mechanisms can occur simultaneously and cancel each other by parity.

Under the Bernoulli-mechanism interpretation above, the probability that no mechanism in $R$ occurs is
\begin{equation}
a := \prod_{e\in R}(1-p_e).
\end{equation}
We call $a$ the \emph{acceptance proxy}.
Equivalently,
\begin{equation}
\Pr\!\left(\exists e\in R:\; e\text{ occurs}\right)
= 1-\prod_{e\in R}(1-p_e)
= 1-a.
\end{equation}
This motivates the \emph{union proxy} for the discard rate,
\begin{equation}
r_{\rm disc}^{\rm UP} := 1-a = 1-\prod_{e\in R}(1-p_e).
\label{eq:rdisc_up_def}
\end{equation}
As shown in \Cref{sec:evaluation}, the union proxy closely tracks the discard rate measured by circuit-level Monte Carlo simulation, $r_{\rm disc}^{\rm MC}$.

\subsection{Acceptance factor decomposition}

To explain why the discard rate varies little across optimized initial configurations, it is useful to separate mechanisms caught immediately by fixed-stabilizer detectors from those caught only by round-parity detectors.
Let $P_{\rm fix}\subseteq P$ denote the set of fixed-stabilizer detectors in the first SE round, and let $|F|:=|P_{\rm fix}|$ be the number of fixed stabilizers.
Then the acceptance proxy factorizes as
\begin{align}
a
&= \prod_{e:\,D_e\cap P\neq \emptyset}(1-p_e) \\
&=
\underbrace{\prod_{e:\,D_e\cap P_{\rm fix}\neq \emptyset}(1-p_e)}_{a_{\rm fix}}
\cdot
\underbrace{\prod_{\substack{e:\,D_e\cap P\neq \emptyset\\ D_e\cap P_{\rm fix}=\emptyset}}(1-p_e)}_{a_{{\rm par}\mid{\rm fix}}}.
\label{eq:acceptance_decomp}
\end{align}
Here $a_{\rm fix}$ captures the mechanisms that would be flagged directly by fixed-stabilizer detectors, while $a_{{\rm par}\mid{\rm fix}}$ captures the remaining postselection-relevant mechanisms that do \emph{not} trigger any fixed-stabilizer detector and are therefore caught only by round-parity detectors.
The two underlying mechanism sets are disjoint, so the product factorization is exact.
Substituting \Cref{eq:acceptance_decomp} into \Cref{eq:rdisc_up_def} gives
\begin{equation}
r_{\rm disc}^{\rm UP}(P)
=
1-a_{\rm fix}a_{{\rm par}\mid{\rm fix}}.
\label{eq:rdisc_up_factored}
\end{equation}

\subsection{Explanation of similar discard rate across different initial configurations}
\label{sec:explanation-similar-discard-rate}

\begin{table}[htb]
    \centering
    \begin{tabular}{ccccccc}
    \toprule
    $p$ & $|\mathcal{I}|$ & $|F|$    & $(a_{\rm fix},\, a_{{\rm par}\mid{\rm fix}})$ & $r_{\rm disc}^{\rm UP}(p)$ & $r_{\rm disc}^{\rm MC}(p)$ \\
    \midrule
$10^{-4}$ & 1 & 26 & (0.9563, 0.8018) & 0.2333 & 0.2331  \\
$10^{-4}$ & 2 & 32 & (0.9487, 0.8062) & 0.2351 & 0.2352  \\
$10^{-4}$ & 3 & 31 & (0.9495, 0.8056) & 0.2351 & 0.2346  \\
$10^{-4}$ & 4 & 32 & (0.9484, 0.8061) & 0.2355 & 0.2352  \\
$10^{-4}$ & 5 & 28 & (0.9517, 0.8040) & 0.2348 & 0.2346  \\
$10^{-4}$ & 6 & 18 & (0.9650, 0.7965) & 0.2314 & 0.2315  \\
\midrule
$5 \times 10^{-4}$ & 1 & 26 & (0.7996, 0.3313) & 0.7351 & 0.7344  \\
$5 \times 10^{-4}$ & 2 & 32 & (0.7683, 0.3407) & 0.7383 & 0.7380  \\
$5 \times 10^{-4}$ & 3 & 31 & (0.7716, 0.3393) & 0.7382 & 0.7383  \\
$5 \times 10^{-4}$ & 4 & 32 & (0.7672, 0.3404) & 0.7388 & 0.7387  \\
$5 \times 10^{-4}$ & 5 & 28 & (0.7808, 0.3360) & 0.7376 & 0.7382  \\
$5 \times 10^{-4}$ & 6 & 18 & (0.8367, 0.3205) & 0.7318 & 0.7316  \\
\midrule
$10^{-3}$ & 1 & 26 & (0.6394, 0.1097) & 0.9298 & 0.9300  \\
$10^{-3}$ & 2 & 32 & (0.5903, 0.1160) & 0.9315 & 0.9312  \\
$10^{-3}$ & 3 & 31 & (0.5954, 0.1151) & 0.9315 & 0.9311  \\
$10^{-3}$ & 4 & 32 & (0.5886, 0.1159) & 0.9318 & 0.9318  \\
$10^{-3}$ & 5 & 28 & (0.6096, 0.1129) & 0.9312 & 0.9313  \\
$10^{-3}$ & 6 & 18 & (0.7001, 0.1027) & 0.9281 & 0.9278  \\
    \bottomrule
    \end{tabular}
    \caption{For each physical error rate $p$ and injection configuration $|\mathcal{I}|$, we report the Monte Carlo discard rate $r_{\rm disc}^{\rm MC}$, the union proxy $r_{\rm disc}^{\rm UP}$, and the two acceptance factors $a_{\rm fix}$ and $a_{{\rm par}\mid{\rm fix}}$ defined in \Cref{eq:acceptance_decomp}. The factorization \(r_{\rm disc}^{\rm UP}=1-a_{\rm fix}a_{{\rm par}\mid{\rm fix}}\) is given by \Cref{eq:rdisc_up_factored}.}
    \label{tab:discard_components}
\end{table}

\Cref{tab:discard_components} shows that the union proxy tracks the Monte Carlo discard rate very closely for all optimized BB code configurations we tested.
It also explains why the discard rate changes only weakly across these configurations, even though the number of fixed stabilizers $|F|$ varies substantially.
From
\begin{equation}
    r_{\rm disc}^{\rm UP} = 1 - a_{\rm fix} a_{{\rm par}\mid{\rm fix}},
\end{equation}
we obtain
\begin{equation}
    \frac{\partial r_{\rm disc}^{\rm UP}}{\partial a_{\rm fix}} = -a_{{\rm par}\mid{\rm fix}}.
\end{equation}
Thus, for a given change $\Delta a_{\rm fix}$,
\begin{equation}
|\Delta r_{\rm disc}^{\rm UP}| \approx a_{{\rm par}\mid{\rm fix}}\,|\Delta a_{\rm fix}|,
\end{equation}
so the sensitivity of the discard rate to changes in $a_{\rm fix}$ is itself controlled by $a_{{\rm par}\mid{\rm fix}}$.

A representative comparison is provided by the optimized configurations with $|\cI|=5$ and $|\cI|=6$.
At high physical error rate, for example $p=10^{-3}$, the change in $a_{\rm fix}$ is relatively large:
\begin{equation}
\Delta a_{\rm fix} \approx 0.7001 - 0.6096 \approx 0.091.
\end{equation}
However, in this regime $a_{{\rm par}\mid{\rm fix}}$ is small, about $0.10$, so the resulting change in discard rate is close to $10^{-2}$.
By contrast, at low physical error rate, for example $p=10^{-4}$, the coefficient $a_{{\rm par}\mid{\rm fix}}$ is much larger, about $0.80$, but the change in $a_{\rm fix}$ becomes much smaller:
\begin{equation}
\Delta a_{\rm fix} \approx 0.9650 - 0.9517 \approx 0.013.
\end{equation}
Again, the resulting change in discard rate remains only of order $10^{-2}$.

The same qualitative mechanism appears across the optimized configurations in \Cref{tab:discard_components}.
At larger $p$, round-parity-only mechanisms are so prevalent that $a_{{\rm par}\mid{\rm fix}}$ is small, which suppresses the impact of variations in $a_{\rm fix}$.
At smaller $p$, $a_{{\rm par}\mid{\rm fix}}$ is larger, but the optimized configurations have much closer values of $a_{\rm fix}$.
Therefore, even substantial changes in the number of fixed stabilizers do not translate into equally large changes in the total discard rate.
This explains why the discard rate remains similar across the optimized initial configurations considered in \Cref{sec:evaluation}.

\section{Maximize Protection of Logical Rotated Pauli Observables}
\label{sec:max-protection-injected-supports}

After the sufficient conditions in \Cref{sec:improved-injection-scheme} are satisfied, the remaining design freedom lies in the basis assignments on data qubits whose initialization is not already fixed by the injection construction. We use this freedom to improve the first-round post-selection of the scheme. More precisely, we formulate a mixed-integer linear programming problem that chooses these remaining basis assignments so as to maximize the number of qubits on the supports of the injected logical observables whose dominant reset faults are detected by fixed stabilizers in the first syndrome extraction round. This objective is a first-order surrogate for reducing the injection error rate: it does not optimize the full logical error probability directly, but it more closely tracks the harmful error mechanisms than simply maximizing the total number of fixed stabilizers.

\subsection{Sets and notation}

Assume we have already settled down the initializations for the qubits in $ \supp(\bar X_i) \cup \supp(\bar Z_i) $,
and the prepared state $ \ket{\psi_{\mathrm{in}}} $ is stabilized by a set of \emph{initial stabilizers} $ \mathcal{R} $,
distinct from the \emph{code stabilizers} $ \mathcal{S} $.
An injectable set $ \cI $ gives a compatible pairing $ \cP = \set{\{u_t,v_t\}}_{t=1}^T $.
So by \Cref{thm:unified-sufficient-condition}, $ \mathcal{R} $ includes
\begin{enumerate}
  \item $ X_{u_t}X_{v_t},\, Z_{u_t}Z_{v_t} $, $ \forall t\in [T] $,
  \item $ M_{Q_i} $, $ \forall i \in \cI $,
  \item $ X_q $, $ \forall q \in A_i \setminus C$, $\forall i \in \cI$,
  \item $ Z_q $, $ \forall q \in B_i \setminus C$, $\forall i \in \cI$.
\end{enumerate}

In the remainder of this section,
we assume we have removed the stabilizers that can never be fixed
due to not commuting with some initial stabilizers in $ \mathcal{R} $.
Then, \( m_X \) and \( m_Z \) are the number of remaining $X$-type and $Z$-type stabilizers, i.e., candidates for being fixed stabilizers.

Let $I \subseteq [n]$
denote the set of data qubits on the supports of the injected logical operators that we want to protect from errors (different from $ \cI $ which denotes the set of injected logicals).
Here, let 
\begin{equation}
I = \bigcup_i \, \big((\supp(\bar X_i) \cup \supp(\bar Z_i)) \setminus \{Q_i\}\big).
\end{equation}

Let 
\begin{equation}
U = [n] \setminus I
\end{equation}
denote the set of data qubits whose single-qubit initialization basis we are free to choose, typically outside the supports of injected logicals.

For each $X$-type or $ Z $-type row $r$ in $H_X$ or $H_Z$, define its full support
\begin{equation}
  \begin{aligned}
  \supp^X(r) &:= \bigl\{ j \in [n] \;\big|\; (H_X)_{r,j} = 1 \bigr\}, \\
  \supp^Z(r) &:= \bigl\{ j \in [n] \;\big|\; (H_Z)_{r,j} = 1 \bigr\},
  \end{aligned}
\end{equation}
and its restriction to $U$,
\begin{equation}
  \begin{aligned}
  \supp^X_U(r) &:= \supp^X(r) \cap U, \\
  \supp^Z_U(r) &:= \supp^Z(r) \cap U.
  \end{aligned}
\end{equation}

For each qubit $j \in I \setminus {\bm Q}$ the injection scheme already fixes its single-qubit basis,
\begin{equation}
  b_j \in \{X,Z\},
\end{equation}
meaning that $j$ is initialized in an $X$-eigenstate if $b_j = X$, or in a $Z$-eigenstate if $b_j = Z$.
To prepare a Bell state $ \{u,v\} $ on $ I $ as described in \Cref{sec:improved-injection-scheme},
we first initialize $ u $ and $ v $ in the $ X $ and $ Z $ bases, respectively,
and then do a CNOT from $ u $ to $ v $.
And thus, we set $ b_u = X $ and $ b_v = Z $.

For each $j \in I$, we define detection sets
\begin{align}
  D_X(j) &:= \bigl\{ r \in [m_X] \;\big|\; (H_X)_{r,j} = 1 \bigr\},
  \label{eq:DX}
  \\
  D_Z(j) &:= \bigl\{ r \in [m_Z] \;\big|\; (H_Z)_{r,j} = 1 \bigr\}.
  \label{eq:DZ}
\end{align}
Thus $D_X(j)$ are the $X$-type stabilizers that anticommute with a $Z_j$ error, and $D_Z(j)$ are the $Z$-type stabilizers that anticommute with an $X_j$ error.

\subsection{Error model and protection}

Following \cite{li_magic_2015,lao_magic_2022}, we adopt a simple noise model for single-qubit initialization (reset) errors:
\begin{itemize}
  \item If qubit $j$ is initialized in the $X$ basis ($b_j = X$), its reset error is modeled as a Pauli $Z_j$ applied after a perfect reset.
  \item If qubit $j$ is initialized in the $Z$ basis ($b_j = Z$), its reset error is modeled as a Pauli $X_j$ applied after a perfect reset.
\end{itemize}
An $X$-type stabilizer $s^X_r$ detects a $Z_j$ error if and only if $(H_X)_{r,j} = 1$ (i.e., $r \in D_X(j)$); similarly, a $Z$-type stabilizer $s^Z_r$ detects an $X_j$ error if and only if $(H_Z)_{r,j} = 1$ (i.e., $r \in D_Z(j)$). 

When single-qubit reset is the only allowed gate,
a $ X $-type stabilizer is \emph{fixed}, i.e., having a deterministic measurement outcome, iff
all the support qubits are initialized in the $ X $ basis;
similarly for $ Z $-type stabilizers.
However, we prepare Bell states to resolve the overlaps,
which involves two-qubit gates;
the injection site is prepared in +1 eigenstate of a rotated Pauli observable $ M $, which do not commute with any $X$-type  or $ Z $-type stabilizer.

For the reasons above, we define the fixed stabilizers with respect to the optimization space $ U $:
\begin{itemize}
  \item An $X$-type row $r$ is \emph{fixed with respect to $U$} if every qubit in $\supp^X_U(r)$ is initialized in the $X$ basis.
  \item A $Z$-type row $r$ is \emph{fixed with respect to $U$} if every qubit in $\supp^Z_U(r)$ is initialized in the $Z$ basis.
\end{itemize}

Note that reset errors on the qubits in the Bell state $ \{u,v\} $ are detectable by the fixed stabilizers,
since both the $ Z $ error on the control qubit $ u $
and the $ X $ error on the target qubit $ v $ penetrate through (commute with) the CNOT($ u $, $ v $). That is why we incorporate the Bell state $ \{u,v\} $ into the set $ I $ and set $ b_u = X $ and $ b_v = Z $.

\subsection{MILP: Decision variables}

We introduce three families of binary decision variables.

\paragraph{Basis choice on $U$.}
For each qubit $j \in U$, we define
\begin{equation}
  x_j \in \{0,1\},
\end{equation}
interpreted as
\begin{equation}
  \begin{aligned}
    x_j = 1 \;\Longleftrightarrow\; \text{initialize qubit $j$ in the $X$ basis}, \\
    x_j = 0 \;\Longleftrightarrow\; \text{initialize qubit $j$ in the $Z$ basis}.
  \end{aligned}
\end{equation}

\paragraph{Row-fixed indicators.}
For each $X$-type row $r$ with $\supp^X_U(r) \neq \emptyset$, we introduce
\begin{equation}
  f_r^X \in \{0,1\},
\end{equation}
indicating whether that $X$ stabilizer is fixed with respect to the basis choices on $U$. Similarly, for each $Z$-type row $r$ with $\supp^Z_U(r) \neq \emptyset$,
\begin{equation}
  f_r^Z \in \{0,1\}.
\end{equation}
Rows with $\supp^X_U(r) = \emptyset$ or $\supp^Z_U(r) = \emptyset$ are unaffected by basis choices on $U$ and are treated as outside the optimization.

\paragraph{Protection indicators for targets.}
For each target qubit $j \in I$, we define
\begin{equation}
  p_j \in \{0,1\},
\end{equation}
where $p_j = 1$ means that the reset error on qubit $j$ is detected by at least one fixed stabilizer, and $p_j = 0$ means that $j$ is unprotected.

\subsection{MILP: Constraints}

\paragraph{Fixed stabilizers and basis compatibility.}
If an $X$-type row $r$ is fixed then all uninitialized qubits in its support must be in the $X$ basis; similarly for $Z$-type rows and the $Z$ basis. This yields the implications
\begin{equation}
  \begin{aligned}
    f_r^X = 1 \;\Longrightarrow\; x_j = 1 \quad \forall j \in \supp^X_U(r), \\
    f_r^Z = 1 \;\Longrightarrow\; x_j = 0 \quad \forall j \in \supp^Z_U(r).
  \end{aligned}
\end{equation}
We encode these as the linear inequalities
\begin{equation}
  f_r^X \le x_j,
  \quad
  \forall r,\ \forall j \in \supp^X_U(r),
  \label{eq:fixed-X}
\end{equation}
\begin{equation}
  f_r^Z \le 1 - x_j,
  \quad
  \forall r,\ \forall j \in \supp^Z_U(r).
  \label{eq:fixed-Z}
\end{equation}
These constraints ensure that whenever $f_r^{X}$ for $ f_r^Z $ is 1, the basis choices on $U$ are consistent with that row being fixed.


\paragraph{Protection of target qubits.}
For each target $j \in I$, we want $p_j = 1$ only if there exists at least one fixed stabilizer that detects the relevant reset error. Using the detection sets in \Cref{eq:DX,eq:DZ}:
\begin{itemize}
  \item If $b_j = X$, then the reset error is modeled as $Z_j$. This is detected by $X$-type rows in $D_X(j)$, but only if they are fixed. Thus we require
  \begin{equation}
    p_j \le \sum_{r \in D_X(j)} f_r^X,
    \qquad \text{if } b_j = X.
    \label{eq:protect-X}
  \end{equation}
  If $D_X(j) = \emptyset$, then there is no stabilizer that anticommutes with $Z_j$, and we simply impose $p_j = 0$.
  \item If $b_j = Z$, then the reset error is modeled as $X_j$. This is detected by $Z$-type rows in $D_Z(j)$, but only if they are fixed. Thus we require
  \begin{equation}
    p_j \le \sum_{r \in D_Z(j)} f_r^Z,
    \qquad \text{if } b_j = Z.
    \label{eq:protect-Z}
  \end{equation}
  If $D_Z(j) = \emptyset$, then we again set $p_j = 0$.
\end{itemize}
\Cref{eq:protect-X,eq:protect-Z} ensure that a qubit can be marked as protected ($p_j = 1$) only if at least one corresponding fixed stabilizer is available.

\paragraph{Integrality.}
All variables are binary:
\begin{equation}
\begin{aligned}
  x_j &\in \{0,1\}, \quad \forall j \in U, \\
  f_r^X &\in \{0,1\}, \quad \forall r \text{ with } \supp^X_U(r) \neq \emptyset, \\
  f_r^Z &\in \{0,1\}, \quad \forall r \text{ with } \supp^Z_U(r) \neq \emptyset, \\
  p_j &\in \{0,1\}, \quad \forall j \in I.
\end{aligned}
\label{eq:integrality}
\end{equation}

\subsection{MILP: Objective}

The goal of the optimization is to maximize the number of qubits on the supports of the injected logicals whose reset errors are detected by some fixed stabilizer. This is captured by the objective
\begin{equation}
  \max \sum_{j \in I} p_j.
  \label{eq:objective}
\end{equation}


Together with \Cref{eq:integrality}, the complete optimization problem is:
\begin{equation}
  \label{eq:max-protection-opt}
  \begin{aligned}
    &\max_{\substack{\{x_j\}_{j \in U},\ \{f_r^X\},\\ \{f_r^Z\},\ \{p_j\}_{j \in I}}} \quad
       \sum_{j \in I} p_j
    \\[4pt]
    \quad
      & f_r^X \le x_j,
      && \forall r,\ \forall j \in \supp^X_U(r),
      \\[2pt]
      & f_r^Z \le 1 - x_j,
      && \forall r,\ \forall j \in \supp^Z_U(r),
      \\[4pt]
      & p_j \le \sum_{r \in D_X(j)} f_r^X,
      && \forall j \in I \text{ with } b_j = X,
      \\[2pt]
      & p_j \le \sum_{r \in D_Z(j)} f_r^Z,
      && \forall j \in I \text{ with } b_j = Z.
  \end{aligned}
\end{equation}


Besides maximizing the number of protected qubits,
with simple modifications,
we can also maximize the number of fixed stabilizers:

\begin{equation}
  \begin{aligned}
    \text{maximize } \quad &
      \sum_{\substack{\forall r \\ \supp_U^X(r) \neq \emptyset}} f_r^X
      +
      \sum_{\substack{\forall r \\ \supp_U^Z(r) \neq \emptyset}} f_r^Z
    \\[3pt]
    \text{subject to } \quad &
      f_r^X \le x_j,
      \quad \forall\, r,\; j \in \supp_U^X(r),
    \\[3pt]
    &
      f_r^Z \le 1 - x_j,
      \quad \forall\, r,\; j \in \supp_U^Z(r).
  \end{aligned}
\end{equation}

In \Cref{sec:evaluation},
we compare the two solutions.
It turns out that maximizing the protection yields lower injection error rate.
This is because any error on the support of an injected logical ($ \bar X_i $ or $ \bar Z_i $),
if not detected by any fixed stabilizer,
will \emph{directly} lead to a first-order logical error on the logical rotated Pauli observable $ \bar M_{i} $.
This is well-captured by the idea to maximize the number of protected qubits.


\bibliography{ref_kun,ref_additional}

\end{document}